\documentclass{article}
\usepackage[margin=1in]{geometry}
\usepackage{cite}
\usepackage{amsmath,amsthm,amssymb,amsfonts}
\usepackage{graphicx}
\usepackage{textcomp}
 \usepackage[table,xcdraw]{xcolor}
\def\BibTeX{{\rm B\kern-.05em{\sc i\kern-.025em b}\kern-.08em
    T\kern-.1667em\lower.7ex\hbox{E}\kern-.125emX}}

\usepackage{subcaption}
\usepackage[noend,ruled,lined,linesnumbered]{algorithm2e}
 \SetAlFnt{\small}
\usepackage{comment}
\usepackage{hyperref}

\usepackage{multirow}
\usepackage{makecell}
\usepackage[normalem]{ulem}
\useunder{\uline}{\ul}{}
\usepackage{longtable}
\usepackage{multirow}
\usepackage{hhline}  
\usepackage{array}
\usepackage{enumitem} 
\usepackage{pifont}

\usepackage{tablefootnote}

\usepackage[framemethod=tikz]{mdframed}
\newmdenv[backgroundcolor=cyan!10,
topline=false,
bottomline=false,
rightline=false,
skipabove=\topsep,
skipbelow=\topsep
]{siderules}

\usepackage{wrapfig}   
\newcommand\ceil[1]{\lceil#1\rceil}

\newcommand{\activee}{\textit{active}}

\newcommand{\osci}{\textit{oscillating}}
\newcommand{\lead}{\textit{leader}}

\newcommand{\portent}{\textit{pent}}
\newcommand{\portext}{\textit{pext}}
\newcommand{\forward}{\textit{forward}}
\newcommand{\backtrack}{\textit{backtrack}}

\newcommand{\nextnode}{\textit{next\_node}}

\newcommand{\Fr}{\textit{$F_r$}}
\newcommand{\groupsize}{\textit{gr\_size}}
\newcommand{\groupno}{\textit{gr\_num}}
\newcommand{\grsource}{\textit{gr\_source}}
\newcommand{\sourcenode}{\textit{Source}}
\newcommand{\required}{\textit{required}}
\newcommand{\helper}{\textit{helper}}

\newcommand{\yes}{\texttt{YES}}
\newcommand{\no}{\texttt{NO}}
\newcommand{\done}{\textsc{Done}}

\newcommand{\cd}{\textsc{Location-Aware Dispersion}}
\newcommand{\dpn}{\textsc{Dispersion}}

\newtheorem{theorem}{Theorem}

\newtheorem{lemma}{Lemma}
\newtheorem{corollary}{Corollary}

\begin{document}



\title{Location-Aware Dispersion on Anonymous Graphs}

\author{
Himani$^1$, Supantha Pandit$^1$, 
Gokarna Sharma$^2$ \\\\
   $^1$Dhirubhai Ambani University, Gandhinagar, Gujarat, India\\
  {\it 202221003@dau.ac.in, pantha.pandit@gmail.com}\\
 $^2$Kent State University, Kent, Ohio, USA\\
{\it gsharma2@kent.edu} 
}

\date{}
\maketitle

\begin{abstract}
The well-studied \dpn~problem is a fundamental coordination problem in distributed robotics, where a set of mobile robots must relocate so that each occupies a distinct node of a network. \dpn~assumes that a robot can settle at any node as long as no other robot settles on that node.  In this work, we introduce \cd, a novel generalization of \dpn~that incorporates location awareness: Let $G = (V, E)$ be an anonymous, connected, undirected graph with $n = |V|$ nodes, each labeled with a color $\mathsf{col}(v) \in \mathcal{C} = \{c_1, \dots, c_t\}, t\leq n$. A set $\mathcal{R} = \{r_1, \dots, r_k\}$ of $k \leq n$ mobile robots is given, where each robot $r_i$ has an associated color $\mathsf{col}(r_i) \in \mathcal{C}$. Initially placed arbitrarily on the graph, the goal is to relocate the robots so that each occupies a distinct node of the same color. When $|\mathcal{C}|=1$, \cd~reduces to  \dpn. There is a solution to \dpn~in graphs with any $k\leq n$ without knowing $k,n$. 

Like \dpn, the goal is to solve \cd~minimizing both time and memory requirement at each agent. We develop several deterministic algorithms with guaranteed bounds on both time and memory requirement. We also give an impossibility and a lower bound for any deterministic algorithm for \cd. To the best of our knowledge, the presented results collectively establish the algorithmic feasibility of \cd~in anonymous networks and also highlight the challenges on getting an efficient solution compared to the solutions for \dpn.

\end{abstract}


\section{Introduction}\label{sec:introduction}
The well-studied \dpn~problem, first proposed by Augustine and Moses Jr.~\cite{augustine2018dispersion},  asks $k \leq n$ mobile robots, initially placed on the nodes of an $n$-node graph $G=(V,E)$, to autonomously reposition so that each node contains at most one robot. 
In this paper, we introduce and study the \cd~problem, a generalization of \dpn~that incorporates location-awareness. 

\begin{siderules}
{\color{red!70!blue} \cd.}
Let $G = (V, E)$ be a connected, undirected, anonymous graph with $n$ nodes. Each node $v \in V$ has a color $\mathsf{col}(v) \in \mathcal{C} = \{c_1, \dots, c_t\}$, where multiple nodes may share a same color. A set $\mathcal{R} = \{r_1, \dots, r_k\}$ of $k \leq n$ mobile robots is given, with each robot $r_i \in \mathcal{R}$ having a color $\mathsf{col}(r_i) \in \mathcal{C}$. Initially, robots are arbitrarily placed on $G$. The goal is to relocate the robots so that each one occupies a distinct node matching its color, with no two robots sharing a node.
\end{siderules}


This variant models scenarios where robots must be assigned to category-specific tasks or resources. For example, in a multi-company charging infrastructure, nodes represent charging stations belonging to different providers, and robots (or electric vehicles) must be routed to stations of the corresponding provider. More generally, \cd~captures coordination in logistics, transportation, and multi-robot systems, where robots are constrained by category-specific requirements. When $|\mathcal{C}|=1$, \cd~reduces to \dpn, as a robot can settle at any node as long as no other robot settles there.

The objective in \dpn~is to place the robots on distinct nodes; once a collision-free (dispersed) configuration is achieved, the problem is considered solved. In contrast, a dispersed configuration alone does not suffice for \cd, i.e., even when all robots occupy distinct nodes, the configuration may still be invalid if a robot is placed on a node whose color does not match its own. Therefore, for \cd, a dispersed configuration, which is a solution to \dpn, is only a necessary condition, but not a sufficient one. 
Additionally, it has been shown that \dpn~is always solvable on any arbitrary $n$-node graph for any number of robots $k \leq n$, without prior knowledge of $k$ or $n$. In contrast, we will show that \cd~is not always solvable. This means \cd~is a more difficult problem than \dpn. 
Consequently, \cd~introduces additional challenges compared to \dpn, since robots must not only avoid collisions but also respect color constraints. As in \dpn, two complexity measures are central: {\em time complexity}, the number of rounds, 
and {\em memory complexity}, the amount of persistent memory per robot. The graph nodes are anonymous and memory-less. 
The objective is to design both time and memory efficient deterministic algorithms for \cd~as well as establish impossibility and lower bounds, 
that characterize the minimal requirements for solving the problem.

\paragraph{Model.}
%
%
We consider an undirected, unweighted graph $G=(V,E)$ with $n=|V|$ nodes, $m=|E|$ edges. A node $v \in V$ is assigned a color $c_i$, $\mathcal{C}=\{c_1,c_2,..., c_t\}, t\leq n,$ is the set of all colors used for the nodes. Each node $v\in V$ is memory-less, meaning it does not store any information. Each edge incident to a  $v$ is labeled using the port number $[1,\delta_v]$, where $\delta_v$ is the degree of node $v$. Hence, an edge connecting two different vertices is assigned a port number at each end. There is no relation between the port labels at the two ends of an edge.
We consider $k\leq n$ robots placed at the node of $G$ denoted by $\mathcal{R}=\{r_1,r_2.....r_k\}$. The initial arrangement of the robots on $G$ corresponds to: (i) the \emph{rooted configuration}, where all $k \leq n$ robots are collocated at a single source node; (ii) the \emph{dispersed configuration}, where each node initially hosts at most one robot; and (iii) the \emph{general configuration}, where the initial placement is neither rooted or dispersed. Each robot has a unique ID in the range $[1,k^{O(1)}]$. If the robots are aware of the type of their initial configuration (rooted, dispersed, or general), the configuration is said to be \emph{known}; otherwise, it is referred to as an \emph{unknown} configuration.
A robot $r \in \mathcal{R}$ is assigned a color $c_i$, where $\mathcal{C}=\{c_1,c_2,..., c_t\}$ is the set of all colors used for the robots, and $t\leq k$. Robots can move between vertices via edges; however, they can only occupy vertices and are not allowed to stay on edges.
%
We consider the {\em local communication} model (aka the face-2-face model), i.e.,  
where the robots located at a node $v \in V$ can communicate and can share their data and ID with only the other collocated robots at $v$. 

\begin{table*}[ht!]
\centering
 \resizebox{\textwidth}{!}{
\begin{tabular}{|p{1.8cm}|p{0.65cm}|>{\centering\arraybackslash}m{3.6cm}|
>{\centering\arraybackslash}m{4.8cm}|>{\centering\arraybackslash}m{6.1cm}|}
\hline

\textbf{Cases} & \textbf{Kwn} & \textbf{Rooted} & \textbf{General} & \textbf{Dispersed} \\ 
 \hline

\multicolumn{5}{|c|}{\cellcolor{blue!7}\textsc{\textbf{Special Cases}}} \\ 
\hline

\hline
{\centering $k \leq n$: Tree, Path, Ring}
&
{\centering Nil}
&
{\centering T: $O(n)$,  M: $O(\log (k+\Delta))$\\ Corollary \ref{cor:tree_n}}
&
{\centering T: $O(n)$, M: $O(\log (k+\Delta))$\\ Corollary \ref{cor:tree_n}}
&
{\centering T: $O(n)$, M: $O(\log (k+\Delta))$\\ Corollary \ref{cor:tree_n}}
\\
\hline


\multirow{2}{*}{\centering $k=1$}
& 

{\centering $k,n$}
& 
{\centering T: $\tilde{O}(n^5)$, M: $O(M^{*})$ Thm: \ref{the:single case}}
& 
- 
& 
{\centering T: $\tilde{O}(n^5)$, M: $O(M^{*})$ Thm: \ref{the:single case}}
\\ \cline{2-5}

& 
{\centering $k$}
& 
{\centering Not solvable, Thm: \ref{the:single case}}
& 
- 
& 
{\centering Not solvable, Thm: \ref{the:single case}} \\

\hline

{\centering $k=n$}
& 
{\centering $k,n$}
& 
{\centering T: $O(m)$, M: $O(\log k)$ \\ Thm: \ref{the:rootedgeneral}} 
& 
{\centering  T: $O(n/k\cdot m)$, M: $O(n/k\cdot \log k)$  \\Thm: \ref{thm:general-k=n}} 
& 
{\centering T: $O(\log k + n/k\cdot m)$,\\ M: $O(n/k\cdot \log k)$, Thm: \ref{thm:dispersed_k=n}} \\ 
\hline

\multicolumn{5}{|c|}{\cellcolor{blue!7}\textsc{\textbf{For Known Configuration}}} \\ 
\hline


\multirow{2}{*}{\centering $2\leq k\leq n$}
& 
{\centering $n$}
& 
{\centering T: $O(n/k \cdot m)$ \\ M: $O(n/k \cdot \log(k+\Delta))$ \\  Thm: \ref{thm:rooted-known}} 
&
{\centering T: $O(\log(n/k)\cdot n/k \cdot m)$ \\ 
M: If $\exists$ source with $>2$ robots: $O(n/k\cdot \log(k+\Delta))$, else: $O(n \cdot \log(k+\Delta))$ \\ 
Thm: \ref{thm:general-unknown}}  
& 
{\centering 
T: (i) $O(n^3 + n/k\cdot m)$, $i=1,2$ \\
(ii) $O(n^i \cdot \log  n+ n/k\cdot m)$, $i={3,4,5}$ \\
(iii) Else: $\tilde{O}(n^{5})$ \\
M: $O(\texttt{$M^{*}$}+n/k \cdot \log(k+\Delta))$\\
Thm: \ref{thm:dispersed_k<n}} \\ \cline{2-5}

& 
{\centering Nil}
& 
{\centering T: $O(\log(n/k)\cdot n/k \cdot m)$, M: $O(n/k \cdot \log(k+\Delta))$, Thm: \ref{thm:rooted-unknown}} 
& 
{\centering Thm: \ref{thm:general-unknown}} 
& 
{\centering Open problem} \\ 
\hline

\multicolumn{5}{|c|}{\cellcolor{blue!7}\textsc{\textbf{For Unknown Configuration}}} \\ 
\hline

\multicolumn{1}{|l|}{\multirow{4}{*}{$2\leq k\leq n$}} 
&
\multicolumn{1}{l|}{\multirow{2}{*}{$n$}} 
& 
\multicolumn{1}{c|}{\makecell[c]{\bf If $k$ known:} Thm: \ref{thm:rooted-known}} 
&
\multicolumn{1}{c|}{\multirow{2}{*}{\makecell[c]{T: (i) $k \ge \lfloor n/2\rfloor+1$: $O(n^3 + n/k\cdot m)$,\\ 
(ii) $k \ge \lfloor n/3\rfloor+1$: $O(n^4 \log n + n/k\cdot m)$,\\ (iii) Else: $\tilde{O}(n^{5})$ \\ 
M: $O(\texttt{$M^{*}$}+n/k \cdot \log(k+\Delta))$\\ Thm: \ref{thm:unknown-config}}}} 
&
\multirow{2}{*}{%
\makecell[c]{%
T: (i) $k \ge \lfloor n/2\rfloor+1$: $O(n^3 + n/k\cdot m)$\\
(ii) $k \ge \lfloor n/3\rfloor+1$: $O(n^4 \log n + n/k\cdot m)$\\
(iii) Else: $\tilde{O}(n^{5})$\\
M: $O(M^{*}+n/k \cdot \log(k+\Delta))$\\ Thm: \ref{thm:unknown-config}
}} \\ \cline{3-3}

\multicolumn{1}{|l|}{}                                 
&
\multicolumn{1}{l|}{}                     
&
\multicolumn{1}{l|}{{\makecell[c]{\textbf{If $k$ unknown}\\  
T: $\tilde{O}(n^{5})$, \\M: $O(\texttt{$M^{*}$}+n/k \cdot \log(k+\Delta))$ \\ 
Thm: \ref{thm:unknown-config}}}} 
&
\multicolumn{1}{l|}{}                      &                       \\ \cline{2-5} 
\multicolumn{1}{|l|}{}                                 & \multicolumn{1}{l|}{\multirow{2}{*}{Nil}} & \multicolumn{1}{l|}{{\makecell[c]{\textbf{If $k$ known:} Thm: \ref{thm:rooted-known}}}  }                                          & 
\multicolumn{1}{c|}{\multirow{2}{*}{\makecell[c]{Open problem}}} 
&
\multirow{2}{*}{\makecell[c]{Open problem}} \\ \cline{3-3}
\multicolumn{1}{|l|}{}                                 
&
\multicolumn{1}{l|}{}                     
& 
\multicolumn{1}{c|}{{\centering \textbf{If $k$ unknown:} Open problem} }                                          
& 
\multicolumn{1}{l|}{}                      &                       \\ \hline
\multicolumn{5}{|c|}{\cellcolor{blue!7}\textsc{\textbf{Lower Bound}}}                                                                                                                                                                                                            \\ \hline
\multicolumn{1}{|l|}{Lower Bound}                                 & \multicolumn{4}{c|}{T: $\Omega(\min\{n^2/k, m\})$, M: $\Omega(n/k\cdot \log k)$, Thm: \ref{thm:lb}} \\ \hline 

\end{tabular}
}
\caption{Proposed results on \cd. \texttt{$M^{*}$} denotes the memory required for the Universal Exploration Sequence (UXS) \cite{molla2023fast}. $i$ is the least distance between robots, $n~(k)$ is the number of nodes (robots), and $\Delta$ is the maximum degree. `-' indicates an impossible configuration and the `Kwn' column indicates whether the robots have prior knowledge of $k$ or $n$; an unspecified entry means that such knowledge may or may not be available, while `Nil' denotes no prior knowledge. In the bounds, `T:' and `M:' denote time and memory per agent, where $\tilde{O}$ suppresses polylogarithmic factors.
}
\label{table:table-result}
\end{table*}

The robots operate in {\em Communicate-Compute-Move} (CCM) cycles; each CCM cycle constitutes a round.  Therefore, in each round, a robot $r\in \mathcal{R}$ at node $v$ will perform these steps: 

{-\em Communicate:} $r$ communicates with the collocated robots at node $v$. 

{-\em Compute:} $r$ performs local computation based on the information shared during the communication phase and decide on whether to leave $v$, which port to use to leave, and what information to write at the robots staying at $v$. 

{-\em Move:} In this phase, $r$ moves through the computed port at $v$.

\begin{table}[ht!]
\begin{center}

\begin{tabular}{|p{2.4cm}|>{\centering\arraybackslash}m{3.3cm}|p{4.6cm}|p{4.5cm}|}
\hline 

\multicolumn{1}{|l|}{\textbf{Configuration}} & \multicolumn{1}{l|}{\textbf{Paper}} & \multicolumn{1}{l|}{\textbf{Time Complexity}} & {\textbf{Memory Complexity}}   \\ \hline

\multicolumn{1}{|l|}{\multirow{3}{*}{Rooted}}  & \cite{sudo2023near} & $O(k)$ & $O(\log k+\Delta)$ \\ \cline{2-4} 
                        
& \cite{sudo2023near}    & $O(k\log k)$    & $O(\log(k+\Delta))$ \\ \cline{2-4} 
                        
& \cite{KshemkalyaniKMP25-SPAA}    & $O(k)$  & $O(\log(k+\Delta))$ \\ \hline

\multicolumn{1}{|l|}{\multirow{2}{*}{General}} & \cite{sudo2023near}  & $O(k\log^2 k)$  & $O(\log(k+\Delta))$ \\ \cline{2-4}

& \cite{KshemkalyaniKMP25-SPAA}    & $O(k)$   & $O(\log(k+\Delta))$ \\ \hline 
Lower bound                 & \cite{augustine2018dispersion}
                                    & $\Omega(k)$              & $\Omega(\log k)$                                  \\ \hline 
                                       
\end{tabular}
\end{center}
\caption{State-of-the-art on \dpn.}
\label{table:dispersion}
\end{table}

\paragraph{Our contributions.}Our main results on \cd~are summarized in Table~\ref{table:table-result}, while a comparison with the state-of-the-art results on \dpn~is provided in Table~\ref{table:dispersion}. We highlight that the lower bound for \cd~is strictly higher than that for \dpn, and moreover, \cd~is impossible to solve when $k=1$ and $n$ is unknown. This stands in contrast to \dpn, which remains solvable under these conditions (see Section~\ref{sec:lower-bound-paragraph}).

We first present a deterministic algorithm for the rooted configuration assuming prior knowledge of both $n$ and $k$. We then extend our results to configurations that are known in advance, namely rooted, dispersed, and general configurations, considering both settings where $n$ and $k$ are known and where such knowledge is unavailable (Sections~\ref{sec:rooted-known} to~\ref{sec:dispersed-known}). Finally, in Section~\ref{sec:unknown-confi}, we study scenarios in which the initial configuration is unknown. Taken together, these results provide a nearly tight
characterization of the complexity of \cd~in anonymous port-labeled memory-less networks, establishing both its algorithmic feasibility and its fundamental limitations.

\paragraph{Related work.} \dpn~is the most closely related problem to \cd.  \dpn~was first introduced by Augustine and Moses Jr. in \cite{augustine2018dispersion}. Later, the problem was studied in more detail for the case $k \leq n$ in \cite{kshemkalyani2019efficient, kshemkalyani2019fast, kshemkalyani2025near, KshemkalyaniKMP25-SPAA, sudo2023near}, considering general graphs when robots have prior knowledge of some parameters, while \cite{shintaku2020efficient} studied the case without such knowledge in the local communication model, where robots at the same node can communicate. The problem has also been explored in other settings, such as dynamic graphs \cite{kshemkalyani2020efficient, agarwalla2018deterministic}, asynchronous models \cite{kshemkalyani2019efficient, agarwalla2018deterministic}, and the global communication model \cite{kshemkalyani2020dispersiona, kshemkalyani2020dispersionb}. Furthermore, dispersion has been studied on square grids \cite{kshemkalyani2020dispersionb, banerjee2025optimal} and triangular grids \cite{himani2024optimal}. 

Compared to \dpn, \cd~is more challenging, as some cases become unsolvable. For the solvable cases, the lower bounds are higher than the \dpn~lower bounds. Furthermore, parameter knowledge as well as initial configuration knowledge plays a significant role on what kind of solutions and guarantees can be achieved. Recently, Sudo {\it et al.}~\cite{sudo2023near} established an important result in the synchronous setting by presenting two algorithms for general initial configurations: (i) a time-optimal $O(k)$ algorithm with memory $O(\Delta + \log k)$, and (ii) an $O(k \log^2 k)$-time algorithm with memory $O(\log(k+\Delta))$. While the first result achieves optimal time, its memory requirement is prohibitively large, and while the second achieves near-optimal memory, it sacrifices time optimality. Kshemkalyani {\it et al.}~\cite{KshemkalyaniKMP25-SPAA} removed these limitations by presenting the first synchronous algorithm that simultaneously achieves both objectives: runs in optimal $O(k)$ rounds while requiring only $O(\log(k+\Delta))$ bits of memory per robot. This yields an $O(\log^2 k)$ factor reduction in time compared to best-known time-optimal result of~\cite{sudo2023near}, thereby closing the gap between time and memory optimality.

 The results of this work are summarized in Table~\ref{table:dispersion}. We study the \cd~problem under rooted, dispersed, and general configurations, considering different assumptions on the robots knowledge of $n$ and $k$. We further analyze \cd~when the initial configuration is unknown.

\section{Impossibility, Lower Bound, and Special Cases}\label{sec:lower-bound-paragraph}

In this section, we investigate fundamental limitations and performance guarantees of the \cd~problem under deterministic robot models. We begin by establishing an impossibility result, showing that in certain settings the problem cannot be solved regardless of the amount of local memory available to the robots. We then derive a lower bound on the time and memory requirements that any deterministic solution must satisfy in the worst case. Finally, we identify and analyze several special cases, both in terms of graph structure and initial configurations, where these negative results can be circumvented, and efficient deterministic algorithms become possible. In particular, we present optimal or near-optimal solutions for trees, paths, and rings, as well as for rooted configurations in arbitrary graphs.


\paragraph{Impossibility} We show that under certain assumptions, the \cd~problem admits no deterministic solution.

\begin{theorem}[{\bf Impossibility}]
\label{the:single case}
Let $G=(V,E)$ be an anonymous, port-labeled, undirected graph with $n$ nodes. For a single robot ($k=1$),  if $n$ is unknown, the \cd~problem is unsolvable, even with unbounded local memory. But, if $n$ is known, the \cd~problem is solvable in $\tilde{O}(n^5)$ time.  
\end{theorem}

\begin{proof}
If the robot knows $n$ \cite{molla2023fast}, it can compute and execute a Universal Exploration Sequence of length $\tilde{O}(n^5)$, which guarantees complete exploration of any anonymous, port-labeled graph with $n$ nodes, allowing it to solve \cd. If $n$ is unknown \cite{sudo2010agent}, a deterministic single robot in an anonymous graph cannot distinguish certain graph configurations and may be trapped in an infinite loop, repeatedly visiting the same nodes while unexplored nodes exist elsewhere. Consequently, it cannot guarantee complete exploration, even with unbounded local memory and the problem is unsolvable without prior knowledge of $n$.
\end{proof}

\paragraph{Time and memory lower bounds of \cd} We now establish our lower bound.

\begin{theorem}[{\bf Lower bound}]\label{thm:lb}
There is an initial configuration of $k\leq n$ robots on a $n$-node graph for which any deterministic algorithm for \cd~requires 
$\Omega(\min\{n^2/k, m\})$ rounds and $\Omega(n/k \cdot \log k)$ bits per robot.
\end{theorem}

\begin{proof} Consider a graph $G$ with $n$ nodes constructed as follows. Partition the nodes into $n/k$ groups (assume $n/k$ is an integer), each containing $k$ nodes. Denote the nodes in group $\ell$ as $v_{\ell k + 1}, v_{\ell k + 2}, \dots, v_{\ell k + k}$, for $0 \le \ell \le n/k - 1$. Each group forms a complete graph on $k$ nodes. Within each group, remove the edges $(\ell k + 1, \ell k + k)$ and $(\ell k + 2, \ell k + k - 1)$. Then, connect consecutive groups using bridge edges $(\ell k + k, \ell k + k + 1)$ and $(\ell k + k - 1, \ell k + k + 2)$.

Assume there is a single red robot at $v_1$, and all nodes are blue except $v_n$, which is red. To achieve correct colorful dispersion, the robot must distinguish between different groups and track which nodes have already been visited.  

Since there are $n/k$ groups, each group must be uniquely identifiable, which requires at least $\log(n/k)$ bits to encode the group identifiers. Within each group, the robot must deterministically explore the $k$ nodes to ensure correct placement and avoid revisiting nodes unnecessarily. At each step, the robot must select the correct port to reach unvisited nodes and eventually navigate the bridge edges to the next group. To encode its current position and make deterministic decisions within a group of $k$ nodes, the robot requires an additional $\log k$ bits of memory.

Combining these inter-group and intra-group requirements, each robot needs at least 
$\Omega\bigl((n/k) \cdot \log k\bigr)$ bits of memory to correctly maintain group information and ensure deterministic colorful dispersion on arbitrary graphs.
\end{proof}





\paragraph{Efficient algorithms for special cases}
We now discuss two special cases B1 and B2. B1 deals with special graphs: tree, path, and rings, whereas B2 deals with rooted initial configuration in arbitrary graphs with known $k,n$ and $k=n$. 

\begin{description}[leftmargin=0em]
    \item[B1:] {\it Tree, Path, Ring: Rooted, General or Dispersed} Configuration, $k\leq n$ (Algorithm \ref{alg:tree_path}):\label{algo:tree} We run a DFS-based strategy. Each robot maintains two states: \textit{explore} and \textit{settle}. At each node $v$, the minimum ID robot whose color matches the node color transitions to the \textit{settle} state, while all other robots in the \textit{explore} state exit the node through port-$(\portent + 1) \bmod \delta_v$, where $\portent$ is the port by which the robot entered $v$ and $\delta_v$ is the degree of $v$. This ensures systematic traversal without revisits, yielding a solution in $O(n)$ rounds. Initially, all robots start at one or more source nodes. In general, and in dispersed configurations, each group of robots, or a single robot, performs its DFS exploration independently. Robots that have settled on a node remain there, preventing conflicts, while other robots ignore occupied nodes and continue their exploration. This guarantees that at most one robot settles per node, achieving \cd. Each robot needs $O(\log (k+\Delta))$ bits to store its state, color, and port.

\begin{corollary}
\label{cor:tree_n}
For any arbitrary initial configuration in a tree, path, and ring, \cd~can be solved in $O(n)$ rounds using $O(\log (k+\Delta))$ bits per robot.
\end{corollary}
\begin{proof}
    The proof directly follows from the standard DFS traversal. At each node $v$, a robot $r$ settles if and only if $\mathsf{col}(r) = \mathsf{col}(v)$ and $r$ has the minimum ID among the eligible robots. For a multi-source configuration, each source initiates its own DFS traversal. When a traversal encounters an already visited node, it continues its DFS without settling. If the node is unvisited, the robot with the minimum ID whose color matches the node settles there, and the remaining robots continue the traversal.
    
    Consequently, the dispersion process completes in $O(m)$ rounds. For a tree or a path, $m = n-1 = O(n)$, and for a ring, $m = n = O(n)$. Thus, in all these cases, the problem can be solved in $O(n)$ rounds. The memory used by the robots to store the variable is $O(\log (k+\Delta))$ bits.
\end{proof}

\begin{algorithm}[H]
\label{alg:tree_path}
\caption{\textsc{\cd\_Tree\_Path\_Ring($r$)} // For $k\leq n$ robots}
    \tcp{$R[v]$ = robots located at node $v$}
    \tcp{Initially $r.p = -1$, where p=\portent} 
    \If{$r.state=explore$}
    {
        \If{$v$ is empty and $\mathsf{col}(v)=c_i$}
        {
            \If{$r \in R[v]$ is the minimum ID robot with $\mathsf{col}(r)=c_i$}
            {
                $r.state=settle$
            }
        }
        \Else{exit through $(p+1)\mod \delta_v$}
    }
\end{algorithm}

    \item[B2:] \textit{Arbitrary Graph, rooted} Configuration, known $n,k$, $k=n$ (Algorithm \ref{alg:rooted_k=n}): \label{algo:arbitrary} For $k=n$ robots starting from a rooted configuration, a DFS-based traversal ensures that each node receives exactly one robot. Each robot maintains a state (\textit{explore} or \textit{settle}) and a phase (\textit{forward} or \textit{backtrack}). In the \forward~phase, robots in the \textit{explore} state move to unvisited nodes. At each node, the minimum-ID robot whose color matches the node color transitions to \textit{settle}, while the others exit via $(\portent + 1) \bmod \delta_v$. When all children of a node are explored, robots enter the \backtrack~phase and return along the entry port, enabling systematic coverage of the graph. Settled robots store parent and child pointers to support forward and backtrack movement. This guarantees a solution to \cd~in $O(m)$ using $O(\log (k+\Delta))$ bits per robot.

\begin{theorem}\label{the:rootedgeneral}
For any rooted initial configuration of $k=n$ robots  in any arbitrary graph with $k,n$ known to robots, \cd~can be solved in 
$O(m)$ rounds using $O(\log k)$ bits per robot.
\end{theorem}
\begin{proof}

By the algorithm, the robots collectively perform a depth-first search (DFS) traversal starting from the source node $S$. At each visited node $v$, if the robot whose color matches $v$ is present, it settles permanently at $v$, while the remaining robots continue the DFS traversal following the rules. Since DFS visits every node of $G$, each node will eventually be reached. Moreover, as $k=n$ and for every color $c_i$ the number of robots with color $c_i$ equals the number of nodes of color $c_i$. Hence, every node is eventually occupied by exactly one robot, and the \cd~problem is solved.  

The DFS explores each tree edge at most twice (once forward and once backtracking) and each non-tree edge at most four times. Thus, the total number of rounds is $O(m)$. Each robot only requires $O(\log k)$ bits to store the variables and its own ID.  
\end{proof}

\begin{algorithm} 

    \caption{\textsc{CD\_Arbitrary\_Rooted\_Graph} // For $k=n$ robots}
    \label{alg:rooted_k=n}

    \tcp{$R[v]$ = robots located at node $v$}

    \tcp{Initially $r.p = -1$, where $p = \portent$}
    
    \If{$r.state=explore$}
    {
    \If{$r.phase=\forward$}
        {
            \If{$v$ is empty and $\mathsf{col}(v)=c_i$}
            {
                \If{$r \in R[v]$ is the minimum ID robot with $\mathsf{col}(r)=c_i$}
                {
                    $r.state$=$settle$ \\
                    store values of $parent(v)$ and $child(v)$
                    
                }
            }
            \Else{$r.phase=backtrack$ and exit through \portent}
        }
        \If{$r.phase=backtrack$}
        {
            \If{$(p+1)\mod \delta_v\neq parent(v)$}
            {
                $r.phase=forward$
            }
        }
            exit through port-$(p+1)\mod \delta_v$
    }
\end{algorithm}

\end{description}

\section{Arbitrary Graph: Rooted, \texorpdfstring{$k\leq n$}{1<k<n}, Known  \texorpdfstring{$n$}{n}}\label{sec:rooted-known}

In this section, we consider the case where $k < n$, with all $k$ robots initially located at a single source node $S$. A traditional implementation of the DFS strategy from Section~\ref{sec:lower-bound-paragraph} (Cases B1 and B2) can cause robots to become trapped in cycles. The robots are constrained by limited memory and cannot store the complete path from the source to their intended destination nodes. During traversal, at each unvisited node, the unsettled robot with the smallest ID whose color matches the node settles there, while the remaining robots continue exploring. If no robot matches the node’s color, all proceed further. In graphs containing cycles of nodes whose colors do not correspond to any robot, the robots may circulate indefinitely, repeatedly revisiting the same nodes without reaching their target nodes. Consequently, a standard DFS strategy under such memory constraints does not guarantee dispersion and may result in deadlocks. To overcome this limitation, we propose a novel DFS-based procedure organized into {\it three} systematic phases, ensuring that each robot eventually reaches a node of matching color and terminates solving \cd. 

We assume $k=\ceil{n/c}$ with $c\leq n$ a constant. Since, in the rooted setting, all robots are collocated at a single source node, which allows them to identify the total number of robots $k$. Thus, whether $k$ is initially known or unknown does not affect the model in the rooted configuration. 

\begin{description}[leftmargin=0em]
    \item[Phase 1: Graph exploration with groups.] The graph nodes are organized into ``\emph{groups}'' of designated sizes, and one robot in each group is designated as the ``\emph{guard}'' (referred to as \lead~or \osci~in algorithm) of that group. The group sizes are not fixed, but are bounded above and below, and are determined based on $k$ and  $n$. Once the DFS traversal is completed, the guard robot at the source node $S$ detects its completion. The algorithm also identifies a set of edges connecting the groups called ``{\it inter-group edges}''. Each group records its parent port (the port leading to its parent group) and child ports (the ports leading to its child groups) w.r.t. these edges. These inter-group edges are guaranteed to be cycle-free. Treating each group as a ``\emph{supernode}'', the groups and inter-group edges collectively form an oriented rooted tree ${\cal T}$, rooted at the supernode containing $S$ (``\emph{root supernode}''). The DFS traversal completes once the root supernode contains at least two robots.

\item[Phase 2: Complete connectivity information gathering.] After Phase 1, the goal is to aggregate robots at $S$ to get the complete connectivity of each node of the graph. The guard robot at root supernode initiates a DFS on ${\cal T}$. Along this traversal, it brings other robots with it, and eventually all robots return to $S$. Each robot stores a portion of the graph structure, and collectively they have the complete connectivity information of nodes in the graph at the source $S$.

\item[Phase 3: Memory-efficient dispersion.] The final task is dispersion. Since robots have connectivity information of each node, that locally determine their destination supernodes (and the specific nodes within them). Each robot must eventually travel from $S$ to its assigned supernode. However, due to memory constraints, we cannot provide each robot with the full path information. To address this, we propose a two-phase mechanism: (i) in the first phase, the path from each target supernode to the root supernode $S$ is established by keeping one robot per supernode; (ii) in the second phase, the required number of robots move toward their designated supernodes along these paths, before finally reaching their exact destination nodes.
\end{description}


\subsection{Overview, Key Challenges, and Techniques}
Several key questions arise in each phase, each presenting challenges. We proceed to address them one by one:

\subsubsection{Phase 1: Graph exploration with groups}

\begin{description}[leftmargin=0em]

\item[Q1:] How does the DFS traversal operate, with groups being formed during the process? (refer to Fig. \ref{fig:example})

In \cd, the strategy follows DFS, with the additional requirement of forming groups of nodes of specified sizes (see Q2) during the traversal. Groups are formed incrementally as the DFS progresses. In the \forward~phase, robots move collectively along unexplored edges, assigning the lowest-ID robot to settle within each newly created group (see Q3), 
and includes the unvisited nodes in the group as the traversal proceeds (see Q2). As exploration proceeds, the leader records parent–child relationships and expands the group until its size reaches valid threshold (see Q3). Once a group is formed, the robot corresponding to that group remains in the group, while remaining robots continue traversal to create subsequent groups.

During the traversal, if a robot encounters a visited node i.e., belongs to an existing group, it either \backtrack s~along explored edges to continue the DFS, or the existing group and current group may dynamically shrink or expand according to the exploration process (see Q5). Whenever a group reaches the required size, a new group is initiated (see Q4). The traversal completes once all nodes  have been explored and at least one robot has returned to the source node (see Q7).

\item[Q2:] What is the group structure?

Each group consists of two types of nodes: \required~and \helper: \required~nodes form the core of the group, as they must ultimately be occupied by robots of matching color in the final dispersion, and \helper~nodes are additional nodes included only to preserve the connectivity of the group; they typically lie along simple paths connecting pairs of \required~nodes. Each group designates a \required~node as its source node. The subgraph induced by a group is always connected, and can be represented as a tree rooted at the group’s source node. The guard robot associated with a group maintains the list of \required~and \helper~nodes. A crucial constraint is that each \required~node belongs uniquely to exactly one group, whereas \helper~nodes serve as \required~nodes for some other group, ensuring non-overlapping coverage of \required~nodes across groups. Additionally, each group maintains information about its inter-group edges, which define the parent–child relationships among groups (see Q6). During Phases 2 and 3, each group can be abstracted as a single \emph{supernode}\footnote{With slight abuse of terminology, we use the terms \emph{supernode} and \emph{group} interchangeably, as they denote the same object with different interpretations. The terminology inter-group edges remains unchanged.}. Consequently, the structure formed by treating groups (or supernodes) as vertices and inter-group edges as connections between them constitutes an oriented rooted tree ${\cal T}$.

\item[Q3:] What is the size of a group?

Formally, let $c = \Theta(n/k)$, and assume $k= \lceil n/c \rceil$ robots with limited memory are initially placed at the source $S$. If at most $\lceil n/c \rceil$ distinct groups are formed, each of size between $2c$ and $4c$ \required~nodes, and exactly one robot is assigned as the guard of each group, then complete exploration of the graph is guaranteed: every node of $G$ belongs to exactly one group as a \required~node, and the collective memory of the robots suffices to store connectivity information about all $n$ nodes. 

If group sizes are restricted to $[c, 2c]$, at most $k$ distinct groups can form; however, no robot may return to the source, preventing DFS completion (see Q8). To avoid this, we consider groups of size between $2c$ and $4c$. A group becomes valid upon reaching $2c$ \required~nodes, ensuring enough robots to guard all nodes and maintain exploration invariants, while exceeding $4c$ triggers a split or restructuring. Since each group has between $2c$ and $4c$ \required~nodes, at most $2c-1$ \helper~nodes suffice to maintain its tree structure. Thus, each group contains at most $6c$ nodes in total. This guarantees that traversal within a group by its guard robot takes only $O(c)$ rounds (see Q4).





\item[Q4:] How does a robot determine, upon reaching a node, whether the node has already been explored? (refer to Fig. \ref{fig:d})

Each group is a rooted tree and can have at most $6c$ nodes (at most $4c$ \required~and at most $2c-1$ \helper~nodes). The guard robot oscillates between the nodes within its group, and thus it requires at most $12c$ rounds to visit every node in that group at least once. Consequently, any robot $r$ that is not a guard robot at a node $v$ waits for $12c$ rounds. After $12c$ rounds, if the robot $r$ encounters the guard robot corresponding to the node $v$ (\required~node for the guard robot) or $v$ is \required~node of $r$, it learns that $v$ has already been visited; otherwise, $r$ determines that it is visiting $v$ for the first time.

\item[Q5:] How does the group size shrink or expand during the traversal? (refer to Fig. \ref{fig:d} $\longrightarrow$ Fig. \ref{fig:e})

A guard robot $r$ expands its group during DFS traversal until it contains $4c$ \required~nodes. Except for a special case, once the group size reaches at least $2c$ and $r$ encounters an already visited node, it stops expanding, as the group is considered sufficient. Robot $r$ shrinks its group when, during the \backtrack~phase (of DFS), it meets another guard robot $r'$ at a \required~node $v$ of its group and its group size is less than $2c$. If all children of $v$ have been explored, robot $r$ reduces its group size and transfers information about the \required~nodes in the subtree rooted at $v$. Based on its own group size, robot $r'$ stores information about nodes farther from $v$.

\item[Q6:] How are inter-group edges identified? (refer to Fig. \ref{fig:b} $\rightarrow$ Fig.~\ref{fig:c})

An \emph{inter-group edge} is defined as an edge $(u, v) \in E$ where two \required~nodes $u$ and $v$ are from two different groups. These edges are identified naturally during DFS: when robots attempt to move along an edge into a node that is already claimed by another group by a robot that recently became a guard robot of that group (as in Q1), they mark that edge as an inter-group edge. Each group records the port numbers of such edges as part of its local data. Collectively, these inter-group edges form the parent-child relationships between groups.


\item[Q7:] How do robots detect the end of DFS? (refer to Fig. \ref{fig:g} $\rightarrow$ Fig.~\ref{fig:h})

The completion of the traversal is detected at the source group (the group that contains the source node). The guard robot at the source maintains a counter that records the number of completed explorations through each incident port, which is incremented by 1 when an incident port is explored. Since each port is explored and returned through exactly twice (forward and backward), the guard concludes that the DFS traversal is complete once its counter reaches $\delta_S$, where $\delta_S$ is the degree of the source node. At this point, all groups have been created, all inter-group edges identified, and all nodes of the graph visited, ensuring that the traversal is complete.

\item[Q8:] How do robots maintain synchrony during DFS traversal?

During DFS traversal, robots are divided into three types: (i) those that continue moving along the DFS path, (ii) the guard robot, which oscillates within its group, and (iii) the recent guard robot, which oscillates within its group and is responsible for forming a new group. The recent guard robot also moves with the DFS-traversing robots, as it expands its group by visiting unexplored nodes (maintaining the oscillation as well). Time is divided into intermediate rounds and movement rounds to maintain synchrony. In an intermediate round, robots wait at their current node to verify whether it has been visited (see Q4). During this period, the guard robot (both previous and recent ones) oscillates; if it encounters another robot, it shares information and either returns to its group’s source node or waits at the current node, depending on the type of communication. Other robots continue to wait and perform computations if they encounter an oscillating guard robot. In the movement round, robots proceed with the DFS traversal to the next node. This structured division ensures that all robots remain synchronized during each node visit.
\end{description}

\begin{figure}[ht!]
     \centering
     \begin{subfigure}[b]{0.32\textwidth}
         \centering
         \includegraphics[width=\textwidth]{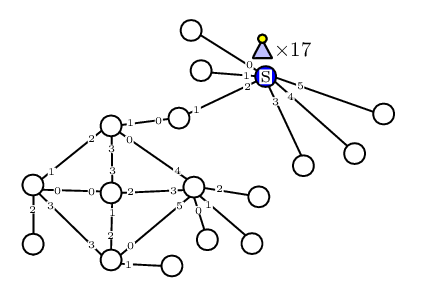}
        
         \caption{}
          \label{fig:a}
     \end{subfigure}
     \hfill
     \begin{subfigure}[b]{0.32\textwidth}
         \centering
         \includegraphics[width=\textwidth]{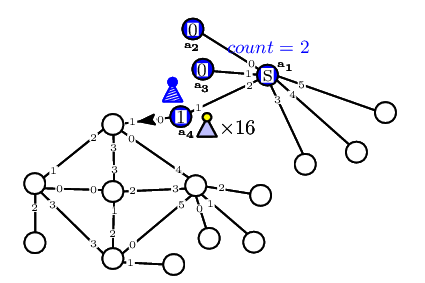}
         
         \caption{}
         \label{fig:b}
     \end{subfigure}
     \hfill      
    \begin{subfigure}[b]{0.32\textwidth}
         \centering
         \includegraphics[width=\textwidth]{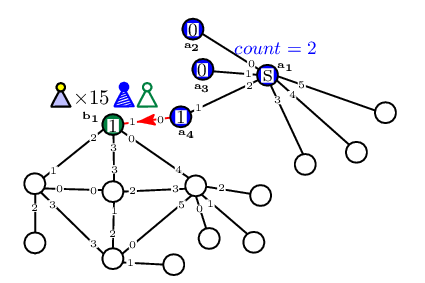}
         
         \caption{}
         \label{fig:c}
        
     \end{subfigure}

     \begin{subfigure}[b]{0.32\textwidth}
         \centering
         \includegraphics[width=\textwidth]{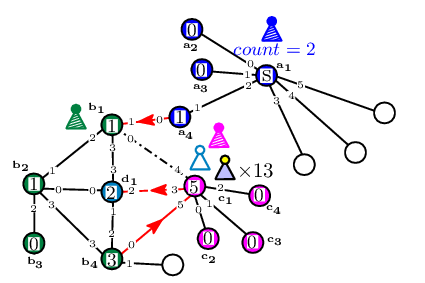}
         
         \caption{}
         \label{fig:d}
     \end{subfigure}
     \hfill      
    \begin{subfigure}[b]{0.32\textwidth}
         \centering
         \includegraphics[width=\textwidth]{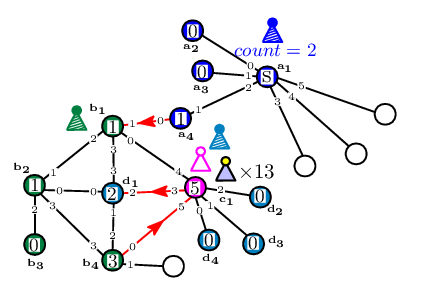}
        
         \caption{}
          \label{fig:e}
     \end{subfigure}
     \hfill
     \begin{subfigure}[b]{0.32\textwidth}
         \centering
         \includegraphics[width=\textwidth]{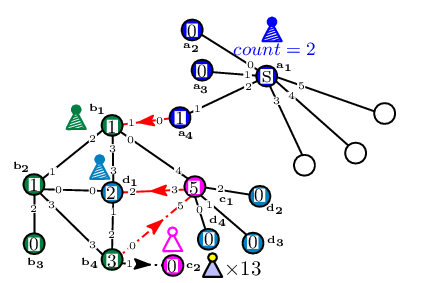}
        
         \caption{}
          \label{fig:f}
     \end{subfigure}
     
    \begin{subfigure}[b]{0.32\textwidth}
         \centering
         \includegraphics[width=\textwidth]{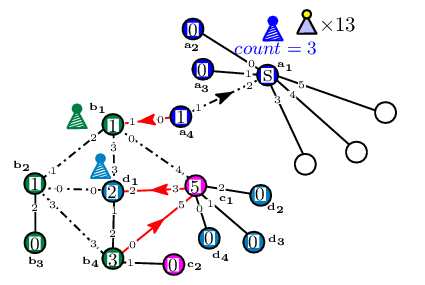}
        
         \caption{}
          \label{fig:g}
         
    \end{subfigure}
     \hspace{1cm}
     \begin{subfigure}[b]{0.32\textwidth}
         \centering
         \includegraphics[width=\textwidth]{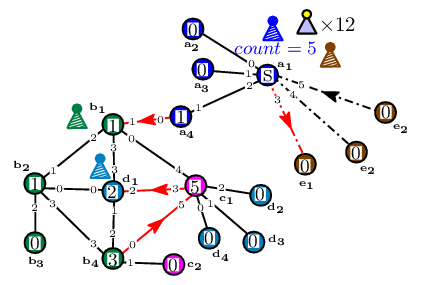}
        
         \caption{}
          \label{fig:h}
     \end{subfigure}

   \caption{Graph $G$, where $hollow$ nodes represent unvisited and $colored$ nodes with labels ($parent$ port) represent visited nodes. Robots with $yellow$ headed denote \activee, $solid$ headed denote \osci, and $hollow$ headed denote \lead. (a) Initial configuration, and begin DFS traversal with \groupsize~between 2 to 4. (b) During traversal, $r_a$ becomes \lead~and expands its group ($a_i$ filled blue, $a_i\in F_{r_a}, i\leq4$), where the black arrow denotes the exit port (including $r_a$, exits its group to store inter-group edges). (c) At $b_1$, \lead~$r_b$ is elected, where red solid line represents inter-group edges with direction from parent to child group. (d) \lead~$r_d$ with $\groupsize<2$, visits an visited node $c_1$, where dotted lines indicate the traversal. (e) $r_d$ expands its group including nodes of subtree at $c_1$, while $r_c$ shrinks, switches to \lead~as $\groupsize<2$. (f) $r_c$ includes a \helper~node $b_4$ and expands its group. (g) \activee~robot traverses back to $S$. (h) End of DFS traversal.}
\label{fig:example}

\end{figure}

\subsubsection{Phase 2: Complete  connectivity information gathering}

Note that after Phase 1, the set of supernodes together with the inter-group edges forms an oriented rooted tree ${\cal T}$, rooted at the supernode containing the source node $S$. Since robots outside the source supernode do not know when Phase 1 terminates, Phase 2 proceeds differently. This phase employs a classical DFS traversal on ${\cal T}$, with two main challenges:

\begin{description}[leftmargin=0em, noitemsep, topsep=0pt]
    \item[(i) Collecting robots at the source:] The guard robot $r$ at $S$ initiates the traversal. During the \forward~phase, $r$ explores new supernodes, guided by the guard robots at those supernodes. When the DFS backtracks from a supernode $W$, the guard robot at $W$ joins the traversal and moves along the backtracking path toward its parent. By the time the traversal completes the subtree rooted at $W$, all robots within that subtree have already been collected and are en route to $S$.

\item[(ii) Maintaining synchronization:] Synchronization is maintained in the same manner as in Phase 1. Specifically, when $r$ reaches a supernode, it spends sufficient time there to ensure that it encounters the guard robot located at that supernode, thereby aligning the traversal across the tree.
\end{description}

\subsubsection{Phase 3: Memory-efficient dispersion} \label{ques:phas3}

After Phase 2, all robots gather at $S$, possessing complete connectivity of all supernodes and nodes therein. Using this information, they collectively perform a local computation followed by a two-stage procedure that ensures every robot reaches its target.

\vspace{1mm}
\textbf{Local computation:} Each robot $r$ determines its {\it destination supernode}, denoted by $dsg\_gr$ (initially \textsc{Null}), and identifies the node within that supernode whose color matches its own. The supernodes $W_1, W_2, \ldots$ are uniquely numbered according to their corresponding group numbers, where $W_i$ represents group $i$. The lowest-ID robot $r$ is assigned to the smallest index $i$ such that $W_i$ contains an unassigned \required~node of the same color, and it sets $r.dsg\_gr = W_i$. This assignment process is then repeated for all robots in increasing order of their IDs. Each robot $r$ also sets a variable $r.d = |R_i|$, where $R_i$ denotes the set of robots whose destination supernode is $W_i$. Once the robot-to-supernode assignment is complete, robots perform subsequent local computations. For each supernode, at most one robot is temporarily assigned to guide the others toward their respective destination supernodes.

Refer to Fig.~\ref{fig:dispersion} for a visual illustration of the computation. The placement of robots is carried out recursively, starting from the leaf supernodes and proceeding toward the source supernode $W_1$ of $\mathcal{T}$, where a robot $r$ assigned to supernode $W_i$ stores its corresponding $r$.{\it level} defined as distance (number of edges) from $W_1$. Each robot $r$ also maintains 3 variables: $stmp\_gr$ (initially \textsc{Null}), which stores the supernode lying on the path from a destination supernode $W_i$ to $W_1$ that contributes to satisfying the demand of $W_i$; $tmp\_gr$ (initially \textsc{Null}), which stores a supernode on the same path used only for traversal purposes; and $ttmp\_gr$ (initially \textsc{Null}), which stores the supernode along the path from $W_i$ to $W_1$ that helps maintain the demand of a node situated on this path. After Stage~1, any robot $r$ with $r.tmp\_gr \neq \textsc{Null}$ returns to $S$.

Let $W_\eta$ be a supernode with child supernodes $W_{\eta_1}, W_{\eta_2}, \dots, W_{\eta_\delta}$. Define $R_\eta$ as the set of robots satisfying $r.ttmp\_gr = W_{\eta_i}$ for some $1 \leq i \leq \delta$; if no such robot exists, then $R_\eta = \emptyset$. Two cases arise:

\begin{enumerate}[noitemsep,leftmargin=1em]
    \item \textbf{$W_\eta$ is a destination supernode:} At least one robot $r$ has $r.dsg\_gr = W_\eta$. The minimum-ID robot $r_m$ among them sets $r_m.stmp\_gr = W_\eta$ and $r_m.ctr = |R_\eta| + r_m.d$, while all others set $r.ttmp\_gr = W_\eta$.

    \item \textbf{$W_\eta$ is not a destination supernode:} If $R_\eta \neq \emptyset$, the minimum-ID robot $r_m$ among those with $r.ttmp\_gr = W_\eta$ sets $r_m.ctr = |R_\eta| - 1$, and all others set $r.ttmp\_gr = W_\eta$. Otherwise, no robot settles at $W_\eta$.
\end{enumerate}

After initial placement, a supernode $W_\eta$ may have no robot $r$ with $r.stmp\_gr = W_\eta$, yet $W_\eta$ lies on the path from one or more destination supernodes $W_i$ to $W_1$. Since the number of robots exceeds the number of supernodes, sufficient robots are available at $W_1$. For each such $W_\eta$ (in order of increasing indices), the lowest-ID robot $r_m$ with $r_m.ttmp\_gr = W_1$ sets $r_m.tmp\_gr = W_\eta$ and resets $r_m.ttmp\_gr = \textsc{Null}$.


Once this computation is completed, each robot $r$ holds information about its current supernode $W_i$, i.e., the non-empty value of either $r.stmp\_gr$ or $r.tmp\_gr$. This is achieved by exchanging information with the robot that originally stored $W_i$'s data. Further, each robot $r$ at a supernode $W_i$ with non-empty $stmp\_gr$ maintains a data structure $D$ that records the demand of robots required by $W_i$ and its child supernodes. Specifically, $D$ tracks the number and colors of robots needed for each child supernode, along with the corresponding group numbers (\groupno) requesting additional robots.

\vspace{1mm}
{\bf Stages:} The robots follow a two-stage cooperative strategy. In Stage 1, they perform a DFS traversal guided by the values stored in $stmp\_gr$ and $tmp\_gr$. Robots remain at nodes corresponding to their $stmp\_gr$, while those with non-empty $tmp\_gr$ return to $W_1$. In Stage 2, starting from the root supernode, robots move according to the data structure $D$ maintained by the robots at each node, which specifies the number and colors of robots required by each child supernode.

\subsection{The Algorithm}

 We begin by defining some key terms and variables that will be used in the description and analysis of the algorithm.
\begin{description}[leftmargin=0em]

 
    \item[Robot types:] There are 3 robot types, based on states: \activee, \lead, and \osci. Initially, all robots are \activee~at source $S$, exploring the graph via DFS. An \activee~robot switches to \lead~to assist in group formation. During this process, the \lead~oscillates within its group. Once the group is formed, it transitions to \osci~and continues oscillating within the group. Furthermore, an \osci~robot may revert to \lead~during group compression.
    \begin{figure}[ht!]
     \centering
         \includegraphics[width=.45\linewidth]{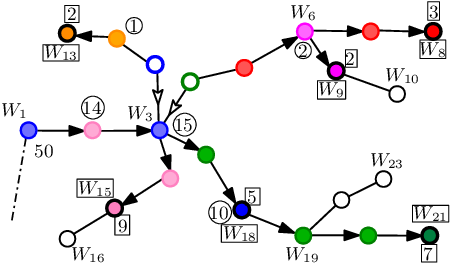} 
        \caption{$W_1$ is the source supernode. Square boxes labeled $W$ denote destination supernodes; the integer inside indicates demand $d$, while circles show the $ctr$ value. Hollow arrows represent the robot’s return to its parent after Stage~1, and solid arrows denote movement from parent to children in Stage~2.}
        

       \label{fig:dispersion}
    \end{figure}

    \item[Group structure:] Let $r$ be a robot and $F_r$ the group it forms. The group consists of 2 types of nodes: \required, key nodes explicitly stored by $r$ in the set $Q_r$, and \helper, which lie along paths connecting \required~nodes and are maintained in the set $H_r$. The subgraph induced by $F_r$ forms a rooted tree with its root at the group’s source node, where $r$ is initially placed.

    Each robot $r$ maintains a data structure $F$ to store information about nodes in $F_r$. The function \textsc{SAVE($r,v$)} (Algorithm \ref{alg:save}) maintains information for a node $v \in F_r$. If $v$ is not already present, a new record is created with default attribute values; otherwise, the existing record is updated. Each record stores the variables mentioned in Table \ref{table:variable}. 
    
    
\begin{table}[h]
\centering
\begin{tabular}{|p{2.6cm}|p{9.7cm}|p{1.3cm}|}
\hline
\textbf{Variable} & \textbf{Description} & \textbf{Initial Value} \\ \hline
$\portent(v)$ & Port through which the robot $r$ entered node $v$ & $-1$ \\ \hline
$\portext(v)$ & Port through which the robot $r$ exited node $v$ & NULL \\ \hline
$parent(v)$ & port through which robot $r$ first visits node $v$, & NULL \\ \hline
$child[v]$ & Stores the list of child ports of node $v$  & $\emptyset$ \\ \hline
$required\_node(v)$ & \yes~if $v$ is a required node; otherwise \no. & \no \\ \hline
$helper\_node(v)$ & \yes~if $v$ is a helper node; otherwise \no. & \no \\ \hline
$node\_num(v)$ & Store the assigned number of node $v$ in the group. It increments by $1$ when a \required~node is added in $F_r$. & $0$ \\ \hline
$Source(v)$ & \yes~if $v$ is a source node of the given graph $G$; otherwise \no. & \no \\ \hline
$gr\_source(v)$ & \yes~if $v$ is the source node of its group; otherwise \no. & \no \\ \hline
$dist(v)$ & Store the distance of $v$ from the group source node; incremented by $1$ in each \forward~step and decremented by $1$ in each \backtrack~step. & $0$ \\ \hline
$gr\_Parent(v)$ & Stores the port of node $v$ through which robot $r$ can reach the parent group according to the traversal & NULL \\ \hline
$gr\_Child[v]$ & Stores the list of ports of node $v$ through which robot $r$ can reach the child groups connected to $F_r$ according to the traversal  & NULL \\ \hline
$\nextnode$ & stores the next node $v$ to be visited (in particular, stored by the \lead~robot), along with the corresponding exit port from the current node $v$, recorded in $\portext(v)$ & NULL \\ \hline
\end{tabular}
\caption{Variables stored in $F$ by a robot $r$ corresponding to a node $v$}
\label{table:variable}
\end{table}

 \begin{algorithm}
    \caption{\textsc{SAVE($r,v$)}}
    \label{alg:save}
    \label{alg:SAVE}
    \tcp{Let robot $r$ is at $v$ and $v \in F_r$}
    \If{$v \notin F[r]$}
    {
        create a new record entry for $v$ in $F[r]$
    }
    \Else
    {
        update attributes of $v$ in $F[r]$
    }
    
\end{algorithm}

\end{description}

\subsubsection{Phase 1: Graph exploration with groups}

The objective of this phase is to explore the entire graph by systematically forming groups such that $k$ robots can visit all the graph nodes, and some robots return to their source after exploring the graph. We will describe the algorithm briefly according to the robot's states: \activee, \lead, \osci.

{\bf \ding{228} Algorithm of \activee~robot (Algorithm \ref{algo:active}):}
    The procedure $Active(r)$ (Algorithm \ref{algo:active}) governs the movement of an \activee~robot $r$ during the exploration phase. Each robot maintains a local variable \groupno~(initially 0), which maintains the number of groups formed by robots, incremented by one whenever a robot initiates a new group formation.  

\begin{description}[leftmargin=0em, noitemsep, topsep=0pt]
    \item[{A.} Initialization (Lines \ref{alg:act1}-\ref{alg:act6}):] In the first round, each \activee~robot increments its \groupno~by 1 (Line \ref{alg:act2}). Among all robots located at the source node $S$, the one with the minimum ID remains at $S$ and changes its state to \lead~(Line \ref{alg:act3}). This robot acts as the ``guard'' of the first group. All remaining robots leave the source node through port-$(p+1)\bmod \delta_S$ (Line \ref{alg:act6}), where $p$ is \portent\footnote{For brevity, $r.\portext(v)$, $r.\portent(v)$ are denoted by $\portext$ and $\portent$, respectively.}~(initially -1) and $\delta_S$ is the degree of the source node. Thus, the exploration begins with a depth-first traversal.

    \item[B. Intermediate rounds (Lines \ref{alg:act7}-\ref{alg:act31})] This last for $(dT+2)$-th to $((d+1)T)$-th rounds, for $d=0, 1,\ldots$, $T = 24c+2$.

      Let the \activee~robot $r_a$ reaches at a node $v$ either in \forward~or \backtrack~phase in the $(dT+1)$-th round. After reaching a node, say $v$, the robot $r_a$ stays at $v$ for the $24c+2$ rounds. During this interval, the robot evaluates its state based on the type of robot encountered at $v$. To determine whether $v$ has been previously visited, $r_a$ waits for $12c$ rounds, allowing the robot (either \osci~or \lead) for which $v$ is a \required~node to arrive at $v$ (as in Q5). Notably, $r_a$ has already waited $12c$ rounds before reaching the decision point, i.e., by the $((dT+2)+12c-1)$-th round it can detect any encounter. The following three cases are possible.

\begin{description}[leftmargin=0em]
\item[(1.)] \textbf{Robot $r_a$ encounters an \osci~robot $r_o$ with $v \in Q_{r_o}$ (Lines \ref{alg:act8}-\ref{alg:act16}):} During the first $12c$ waiting rounds ($(dT+2)$-th to $((dT+2)+12c-1)$-th), if $r_a$ encounters an \osci~robot $r_o$ with $v \in Q_{r_o}$, it means that $v$ is already visited, and its information is stored by $r_o$. There are two possible cases.

\begin{description}[leftmargin=0em]
    \item[(i)] {\bf $r_o.\sourcenode(v)=\yes$ (Lines \ref{alg:act9}–\ref{alg:act10} and \ref{alg:act13}–\ref{alg:act16}):} This case occurs when $r_a$ backtracks or reaches $S$ during traversal. There are two cases: (i) when $r_a$ returns to $S$ after exploring the entire graph (Lines \ref{alg:act9}–\ref{alg:act10}) to complete its task; it stores that DFS is done. (ii) when $r_a$ returns before finishing exploration (Lines \ref{alg:act13}–\ref{alg:act16}) due to the algorithm or graph structure. The \activee~robot then continues exploring the remaining graph.

    \item[(ii)] {\bf $r_o.\sourcenode(v)=\no$ (Lines \ref{alg:act11}-\ref{alg:act16}):} The \activee~robot continues exploration. If $r_a$ is in the \forward~phase, it switches to \backtrack~and sets $\portext=\portent$ (Lines \ref{alg:act11}–\ref{alg:act12}). If $r_a$ is in the \backtrack~phase, it updates $\portext$~(Lines \ref{alg:act13}–\ref{alg:act14}). If $\portext$ value differs from $parent(v)$, indicating that not all neighbors of $v$ have been explored, $r_a$ switches back to \forward~(Lines \ref{alg:act15}–\ref{alg:act16}); otherwise, it remains in \backtrack~phase.
    \end{description}

    \item[(2.)] \textbf{Robot $r_a$ does not encounter any \osci~robot $r_o$ with $v \notin Q_{r_0}$ (Lines \ref{alg:act18}-\ref{alg:act31}):} In this case, there are two cases based on whether a leader $r_l$ is located at $v$ or not:

     \begin{description}[leftmargin=0em]
         \item[(i)] {\bf \lead~$r_l$ is located at $v$ (Lines \ref{alg:act18}-\ref{alg:nact27}):} Two cases arise based on whether the \lead~stores information about $v$ and the phase of the \activee~robot. If the \activee~robot $r_a$ is in the \forward~phase, it performs no action and backtracks from $v$ (Lines~\ref{alg:act18}–\ref{alg:nact20}). If $r_a$ is in the \forward~phase and the leader $r_l$ stores information about $v$, where $v$ is the source node with $r_l.count=\delta_S-1$, then upon returning to $S$ after exploring the graph, $r_a$ sets $r_a.DFS=\done$ (Lines~\ref{alg:nact22}–\ref{alg:nact23}). Otherwise, if the traversal is incomplete or $v$ is not the source, $r_a$ remains idle and stores \portext~(Line~\ref{alg:nact25}). If some ports of $v$ remain unvisited, i.e., $(p+1)\bmod \delta_v \neq parent(v)$, $r_a$ switches to the \forward~phase (Lines~\ref{alg:nact26}–\ref{alg:nact27}); otherwise, it backtracks.

            
    

        \item[(ii)] {\bf no \lead~is at $v$ (Lines \ref{alg:act26}-\ref{alg:act31}):} This occurs when the last \lead~robot completes its group formation and switches to \osci. The \activee~robot $r_{a_m}$ with minimum ID, becomes \lead, records $v$, and initiates a new group (Lines \ref{alg:act27}–\ref{alg:act29}). The others update $\portext$ and increment \groupno. If all ports $v$ are explored, the robots switch to \backtrack~(Lines \ref{alg:act30}–\ref{alg:act31}).

\end{description}
    \end{description}

    \item[C. Waiting and exit (Lines \ref{alg:act32}-\ref{alg:act38}):] The \activee~robot $r_a$ performs computation at $v$ in $((dT+1)+12c+1)$-th round. Other time, i.e., from  $(dT+2)$-th to $((dT+2)+12c-1)$-th and from $((dT+2)+12c)$-th to $((d+1)T)$-th round, the robot waits at $v$ (Line \ref{alg:act38}).  Otherwise, at the conclusion of its stay at node $v$, specifically in the $((d+1)T+1)$-th round, the robot departs through the exit port stored in \portext~(Line \ref{alg:act36}).
    \item[D. End of traversal (Lines \ref{alg:act9}-\ref{alg:act10}, \ref{alg:nact22}-\ref{alg:nact23}, and \ref{alg:act33}):] The robot responsible for the source node $S$ keeps a counter of encounters with \activee~robots. When the counter reaches $\delta_S$, the degree of $S$, all ports of $S$ have been traversed and all returning robots have revisited $S$. At this point, the DFS traversal is complete, and the exploration ends ($r.DFS=\done$). This end is explicitly captured in (Lines \ref{alg:act9}-\ref{alg:act10} and \ref{alg:nact22}-\ref{alg:nact23}). The robots wait at the node and in the $((d+1)T+1)$-th round, change their state to $passive$ and run the next phase (Line \ref{alg:act33}).
\end{description}

       \begin{algorithm}[ht!]
    \caption{$Active(r)$}
    \label{algo:active}
    \tcp{$R[v]$ = robots located at node $v$}
    \tcp{Initially $r.p = -1$, where $p = \portent$ }

    \If{$round==1$ \label{alg:act1}}
    {
         $r.\groupno= r.\groupno+1$ \label{alg:act2}\\
        \lIf{$r$ is the minimum ID robot \label{alg:act3}}
        {
            $r.state=leader$ \label{alg:act4}     }
        \lElse
        {
            exit through $\portext=(p+1)\bmod \delta_S$ \label{alg:act6}
        }
    }

    \If{$round== (dT+1)+12c+1$ \label{alg:act7}}
    {
        \If{$\exists$ \osci~robot $r_o\in R[v]$ with $v \in Q_{r_o}$ \label{alg:act8}}
        {
            \If{$r_o.\sourcenode(v)=\yes$ in $r_o.F$ and $r_o.count=\delta_v-1$ \label{alg:act9}}
            {
                $r.DFS=\done$  \label{alg:act10}
            }
            \ElseIf{$r.phase= \forward$ \label{alg:act11}}
            {
                $r.phase=\backtrack$ and $\portext= p$ \label{alg:act12}
            }
            \ElseIf{$r.phase= \backtrack$ \label{alg:act13}}
            {
                $\portext= (p+1)\bmod \delta _v$ \label{alg:act14}\\
                \If{$(p+1)\bmod \delta _v \neq parent(v)$ \label{alg:act15}}
                {
                    $r.phase=\forward$ \label{alg:act16}
                }
               
            }
        }
        \Else 
        {
            
            \If{$\exists$ \lead~robot $r_l \in R[v]$ with $v \in Q_{r_l}$ \label{alg:act18}}
            {
                \If{$r.phase=\forward$ \label{alg:nact19}}
                {$r.phase=\backtrack$ \label{alg:nact20}}
                \Else{\If{$r_l.\sourcenode(v)=\yes$ in $r_l.F$ and $r_l.count=\delta_v-1$ \label{alg:nact22}}
                {
                    $r.DFS=\done$ \label{alg:nact23}
                }
                \Else
                {
                    $\portext= (p+1)\bmod \delta_v$ \label{alg:nact25}\\
                    \If{$(p+1)\bmod \delta_v \neq parent(v)$ \label{alg:nact26}}
                    {
                         $r.phase=\forward$ \label{alg:nact27}\\
                    }
                   
                }}  
            }
            \Else
            {
                $r.\groupno=r.\groupno+1$ \label{alg:act26}\\
                \If{$r \in R[v]$ is the minimum ID robot \label{alg:act27}}
                {
                    \textsc{Save($r,v$)} and store $gr\_Parent(v)=p'$, where $p'$ is \portent~of \osci~robot with $recent=true$ 
                    and $r.state=leader$  \label{alg:ac28}
                 }
                 $\portext= (p+1)\bmod \delta_v$ \label{alg:act29}\\
                \If{$(p+1)\bmod \delta_v= parent(v)$ \label{alg:act30}}
                {
                     $r.phase=\backtrack$ \label{alg:act31}\\
                }
           
            }
        }
    }
    \If{$round == (d+1)T+1$ \label{alg:act32}}  
    {
        \If{$r.DFS=\done$ \label{alg:act33}}
        {
            $r.state=passive$ and run Phase 2 \label{alg:act34}
        }
        \lElse
        {
            exit through $\portext$ \label{alg:act36}
        }
    }
    \lElse{wait at the current node \label{alg:act38}
    }
\end{algorithm}

{\bf \ding{228} Algorithm of \lead~robot:} When the robot $r$ becomes \lead, it executes $Leader(r)$ (Algorithm~\ref{algo:leader}). It increments the group number ($r.\groupno = r.\groupno + 1$), maintains $r.F$, $Q_r$, $H_r$, and a counter $count$. For each node $v$ added to $Q_r$, it increments \groupsize. The leader $r$ expands the group through the same port (\portext) as the \activee~robot, while \activee~robots continue exploring unvisited nodes to support group formation.

\begin{description}[leftmargin=0em]
   \item[A. Initialization (Lines \ref{alg:lead1}-\ref{alg:lead4}):] In the first round, the robot $r_l$ with the minimum ID becomes \lead~and stores the source node $S$ in $r_l.F$ (Line \ref{alg:lead2}). It increments \groupsize~by 1 (Line \ref{alg:lead4}) and stores the information of the next node to be visited by exiting $S$ through port-$(\portent+1) \mod \delta_S$, where \portent~is the port through which an \activee~robot $r_a$ entered $S$ (Line \ref{alg:lead4}).

    \item[B. Intermediate rounds (Lines \ref{alg:nlead5}-\ref{alg:lead47}):] These last from $(dT+2)$-th to $((d+1)T)$-th rounds, for $d=0,1,\ldots$, $T=24c+2$. It is divided into three parts. From $(dT+2)$-th to $(dT + 2)+12c$-th rounds, the \lead~robot $r_l$ oscillates within its group, and the next upcoming node, stored in the variable $\nextnode$. In $(dT+1)+12c+1$-th round, the \lead~robot do the computation and expand its group. From $(dT+1)+12c+2$-th to $((d+1)T+1)$-th, the \lead~robot oscillates back to its group source node using the stored information in $r.F$.

    If it encounters an \activee~robot, it stops and waits there (Lines \ref{alg:nlead5}-\ref{alg:nlead8}). 

    Let the \lead~robot $r_l$ encounters an \activee~robot at the node $v$ while oscillating during $(dT+2)$-th to $((d+1)T)$-th rounds (Lines \ref{alg:nlead5}-\ref{alg:nlead8}). The robot $r_l$ wait at $v$. In $((dT+1)+12c+1)$-th round, at $v$, following cases arise at $v$:

    \begin{description}[leftmargin=0em]
\item[(1.)] {\bf $r_l$ does not encounter an \osci~robot $r_o$ with $v \in Q_{r_o}$ (Lines \ref{alg:nlead9}-\ref{alg:nlead30}):} There are two cases based on the phase of an \activee~robot. If the \activee~robot is in \backtrack~phase, then there are three possible cases: (i) If $v$ is a source node with $v \in Q_{r_l}$ and $r_l.count = \delta_v - 1$, i.e., $r_l$ has completed the DFS traversal and returns to $S$, it sets $r_l.DFS = \done$ and changes its state to \osci~(Lines \ref{alg:nlead12}-\ref{alg:nlead13}). (ii) Otherwise, if $v$ is a source node and $r_l.count \neq \delta_v - 1$ (Lines \ref{alg:nlead14}-\ref{alg:nlead15}), $r_l$ increments $count$ and records the information of the next node to be visited, along with the exit port from the current node $v$. (iii) If $v$ is not a source node and $v \in Q_r$ (Lines \ref{alg:nlead16}-\ref{alg:nlead17}), $r_l$  records the information of the next node to be visited, along with the exit port from the current node $v$. 

If the \activee~robot is in the \forward~phase, there are two cases. If $v \in Q_{r_l}$ (Line \ref{alg:nlead19}), then $r_l$ does nothing as the \activee~robot backtracks from $v$. If $v \notin Q_{r_l}$ (Lines \ref{alg:nlead22}–\ref{alg:nlead30}), $r_l$ adds the newly visited node $v$ to $F_r$. Now, if $\groupsize < 4c$, $r_l$ continues group expansion; otherwise, when $\groupsize = 4c$, it switches to the \osci~state, and updates \portext~(to store parent-child relationship) based on whether the \activee~robot is in the \forward~or \backtrack~phase.

\item[(2.)] {\bf $r_l$ encounters an \osci~robot $r_o$ with $v \in Q_{r_o}$ (Lines \ref{alg:lead24}-\ref{alg:lead46}):} This occurs when $v$ has already been visited. Let $g$ denote the current \groupsize~of $r_l$, and $g'$ the updated size after expansion. If $v$ is $S$ with $r_o.count=\delta_v-1$ (Lines \ref{alg:lead25}-\ref{alg:lead27}), $r_l$ meets the termination condition, switches to \osci, sets $recent=false$, returns to its group in $((d+1)T+12c+1)$-th round. If $2c \leq g < 4c$ (Lines \ref{alg:lead25}-\ref{alg:lead27}), $r_l$ has enough \required~nodes, it switches to \osci, sets $recent=false$, records the exit port, sets $r.move=1$. If $g < 2c$ (Lines \ref{alg:lead32}-\ref{alg:lead47}), the action depends on the \activee~robot’s phase: in the \forward~phase (Lines \ref{alg:lead32}-\ref{alg:lead33}), $r_l$ store the information of the next upcoming visiting node and exit port as $p$; while in the \backtrack~phase, if an unvisited child exists (Lines \ref{alg:lead35}-\ref{alg:lead36}), it explores that child (stores information in \nextnode, i.e., the next node to be visited) and records $v$ as a \helper~node. Otherwise (Lines \ref{alg:lead38}-\ref{alg:lead46}), it expands the group using descendant information from $r_o$, and maintains $g' \leq 4c$. When $2c \leq g' \leq 4c$, it switches to \osci, otherwise, it adds $v$ in $r_l.F$ and continues expansion until $g' \geq 2c$.

\end{description}

From rounds $((dT+2)+12c+1$)-th to $((d+1)T$)-th (Lines \ref{alg:nlead51}-\ref{alg:nlead52}), $r_l$ oscillates back to its group source node.

\item[C. Waiting and exit  (Lines \ref{alg:lead47}–\ref{alg:lead49}):] At $((d+1)T+1)$-th round, if DFS is complete, it begins the next phase.

\end{description}

{\bf \ding{228} Algorithm of \osci~robot:}  
A robot $r$ in the \osci~state follows $Oscillating(r)$ (Algorithm~\ref{algo:osci}). The behavior depends on the variable $recent$. If $recent = true$, robot $r$ moves with an \activee~robot to collect information about inter-group edges; if $recent = false$, it oscillates within its current group.

\begin{description}[leftmargin=0em]

\item[A. Initialization (Line~\ref{alg:osci1}):]  
In the $(dT+1)$-th round, the robot does nothing.

\item[B. Intermediate rounds (Lines~\ref{alg:osci3}-\ref{alg:osci41}):]  
It is divided into three parts. From the $(dT+2)$-th to the $((dT+2)+12c-1)$-th rounds, if $recent = false$, the \osci~robot $r_o$ oscillates within its group, pausing upon encountering an \activee~or a \lead~robot; if $recent = true$, it stays at its current node. In the $((dT+2)+12c)$-th round, it performs computation through communication with collocated robots. Finally, from the $((dT+2)+12c+1)$-th to the $((d+1)T+1)$-th rounds, the \osci~robot continues oscillation and returns to its group source node.

Let an \osci~robot $r_o$, with $recent = false$, encounter an \active~robot at node $v$ during the $(dT+2)$-th to the $((dT+2)+12c-1)$-th rounds (Lines~\ref{alg:osci3}-\ref{alg:osci7}). The robot $r_o$ waits at the current node $v$, and in round $(dT+2)+12c$ (Lines~\ref{alg:osci10}-\ref{alg:osci36}), the following cases arise:

\begin{itemize}[leftmargin=.5em]

\item $recent = false$, $move = 0$:  
If $v$ is the source node and $r_o.count = \delta_v - 1$ (Lines~\ref{alg:osci12}-\ref{alg:osci13}), i.e., the robot returns after exploring the whole graph, the DFS terminates and $r_o$ sets $DFS = \done$. If $r_o.count \neq \delta_v - 1$, i.e., the robot returns to the source node after exploring a port $p$ (the entered port) due to the structure of the graph, then $r_o$ increments $count$ (Line~\ref{alg:osci15}).

If a \lead~robot $r_l$ is located at $v$, then $r_o$ shares the information of its group. Based on the computation of $r_l$, if it expands its group using the nodes stored by $r_o$, then $r_o$ updates its group size $g$ and updates its state accordingly (Lines~\ref{alg:osci19}-\ref{alg:osci27}). If $g < 2c$, $r_o$ becomes a \lead~robot and stores the information of the next node to be visited (Lines~\ref{alg:osci19}-\ref{alg:osci21}). If the group size satisfies $2c \le g \le 4c$, then $r_o$ remains in the \osci~state and moves with an \activee~robot to a node that is not in its group in order to record the parent-child relationship between groups (Lines~\ref{alg:osci23}-\ref{alg:osci25}). If a \lead~robot is located at $v$, then $r_o$ moves with an \activee~robot in the \forward~phase or stays if the \activee~robot is in the \backtrack~phase (Lines~\ref{alg:osci29}-\ref{alg:osci30}).

\item $recent = true$ (Lines~\ref{alg:osci31}-\ref{alg:osci34}):  
If no \osci~robot is located at node $v$ with $v$ marked as a \required~node, then $r_o$ records the entry and exit ports for new groups (inter-group edges) and sets $recent = false$.

\item $recent = false$, $move = 1$ (Lines~\ref{alg:osci35}-\ref{alg:osci36}):  
The robot $r_o$ exits via port \portext.

\end{itemize}

From rounds $((dT+2)+12c+1)$-th to $(d+1)T$-th (Lines~\ref{alg:osci37}-\ref{alg:osci41}), $r_o$ oscillates back to its group source node if $recent = false$; otherwise, it remains at its current node.

\item[C. Waiting and exit (Lines~\ref{alg:osci42}-\ref{alg:osci46}):]  
The robot $r_o$ with $recent = true$ exits the current node to record inter-group edges (Line~\ref{alg:osci43}). Once $r_o.DFS = \done$, the robot proceeds to the next phase of the algorithm (Line~\ref{alg:osci45}).

\end{description}

 \subsubsection{Phase 2: Complete  connectivity information gathering} \label{des:gathering}

The algorithm begins with \osci~robot $r$ at source node $S$ (or, equivalently, at the source supernode). Robot $r$ transitions to $osci\_head$ state and initiates a DFS traversal on groups. Since every \osci~robot knows its parent, child groups, traversal across groups is realized via inter-group edges.


\begin{algorithm}[ht!]
    \caption{$Leader(r)$}
    \label{algo:leader}
     \tcp{$p'=(p_a+1)\bmod \delta_v$, where $p_a=r_a.\portent$}

    \If{$round==1$ and $r \in R[v]$ is the minimum ID robot \label{alg:lead1}}
    {
        $r.state=\lead$,  and \textsc{Save($r,S$)} \label{alg:lead2}\\
         $r.\groupsize=\groupsize+1$, $r.\nextnode=v$, $\portext(S)=p'$ \label{alg:lead4}
    }
    \If{$round= ((dT + 1), (dT + 1) + 12c]$ \label{alg:nlead5}}
    {
        oscillates within its group nodes and node $v$\\
        \lIf{encounters an \activee~robot $r_a$}
        {
            wait \label{alg:nlead8}
        }
    }

    \If{$round == (dT + 1) + 12c + 1$ \label{alg:nlead9}}       
    {      
        \If{$\nexists$ \osci~robot $r_o \in R[v]$ with $v \in Q_{r_o}$ \label{alg:nlead10}}
        {
            \If{\activee~robot $r_a \in R[v]$ is in \backtrack~phase\label{alg:nlead11} }
            {\If{$r.\sourcenode(v)=\yes$ in $r.F$ and $r.count=\delta_v-1$ with $v \in Q_{r}$  \label{alg:nlead12}}
            {
                $r.count=r.count+1$, $r.state=\osci$, $r.DFS=\done$ \label{alg:nlead13}
            }
            \ElseIf{$r.\sourcenode(v)=\yes$ in $r.F$ and $r.count\neq \delta_v-1$ with $v \in Q_{r}$  \label{alg:nlead14}}
                {
                    $r.count=r.count+1$, $r.\nextnode=v'$ and $\portext(v)=p'$ \label{alg:nlead15}
                }
            
            \ElseIf{ $v \in Q_{r}$ \label{alg:nlead16}}
            {$r.\nextnode=v'$ and $\portext(v)=p'$\label{alg:nlead17}}
            }
            
            \ElseIf{\activee~robot $r_a \in R[v]$ in \forward~phase \label{alg:nlead18}}
                {
                    \lIf{$v \in Q_r$ \label{alg:nlead19}}
                    {do nothing \label{alg:nlead20}}
                    \Else{
                    \textsc{Save($r,v$)}, and
                    $r.\groupsize=r.\groupsize+1$ \label{alg:nlead22}\\
                    \If{$r.\groupsize<2c$ \label{alg:nlead23}}
                    {
                       $r.\nextnode=v'$ and $\portent(v')=p'$ \label{alg:nlead24}
                    }
                
                    \Else
                    {
                        $r.state=\osci$ \label{alg:nlead26}\\
                        \If{$p' \neq parent(v)$\label{alg:nlead27}} 
                        {
                            $recent=true$ and $\portext= p'$ \label{alg:nlead28}
                        }
                        \lElse
                        {
                            $recent=false$ \label{alg:nlead30}
                        }
                    }
                }
            }

            }

        \If{$\exists$ \osci~robot $r_o \in R[v]$ with $v \in Q_{r_o}$ \label{alg:lead24}}
        {

            \If{$r_o.\sourcenode=\yes$ and $r_o.count=\delta_v-1$ or $r.\groupsize \in [2c,4c)$  \label{alg:lead25}}
            {
                 $r.move=1$, $r.recent=false$, $r.state=\osci$ and $\portext=p_a$ \label{alg:lead27}
            }
                 
            \Else
            {
                \If{\activee~robot $r_a \in R[v]$ in \forward~phase \label{alg:lead32}}
                {
                    $r.\nextnode=v'$ and $\portent(v)=p_a$ \label{alg:lead33}
                }
                \Else
                {
                    \If{$p' \neq parent(v)$ \label{alg:lead35}}
                    {
                       $r.\nextnode=v'$ and $\portext(v)=p'$ \label{alg:lead36}
                    }
                    \Else
                    {
                        $r$ expand its group \label{alg:lead38}\\
                        \If{after expand $r.\groupsize= [2c,4c]$ \label{alg:lead39}}
                        {
                            $r.state=\osci$, $r.recent=false$ \label{alg:lead40}
                        }
                        \Else
                        {
                            \textsc{Save($r,v)$}, $r.\groupsize= r.\groupsize +1$ \label{alg:lead42}\\
                             
                            \If{$r.\groupsize < 2c$ \label{alg:lead43}}
                            {
                                $r.\nextnode=v'$, $\portext(v)=p'$ \label{alg:lead44}
                            }
                            \Else
                            {$r.state=\osci$ and $r.recent=false$\label{alg:lead46}
                            }
                        }
                    }
                }     
            }
        }
        }
    \If{$round= ((dT + 1) + 12c + 1, (d + 1)T + 1)$ \label{alg:nlead51}}
    {
        oscillate and return to \grsource~node \label{alg:nlead52}
    }

    \If{$round==(d+1)T+1$ \label{alg:lead47}}
    {
        \lIf{$r.DFS=\done$ \label{alg:lead48}}
        {
            run Phase 2 \label{alg:lead49}
        }
    }
   
\end{algorithm}

The procedure proceeds in synchronous rounds. During exploration, an \osci~robot oscillates within its group: it moves from its source node to another designated node during $(dT+2)$ to $(dT+2)+(12c-1)$, pauses for one round, and returns to the source during $(dT+2)+(12c)$ to $(d+1)T$.

Similarly, the $osci\_head$ $r$ moves from its source node in $F_r$ to a child group $F_{r'}$ along an inter-group edge during $(dT+2)$ to $(dT+2)+(12c-1)$, reaching a node $v \in F_{r'}$. Upon arrival, $r$ remains at $v$ until round $(d+1)T$, ensuring that \osci~robot $r'$ of $F_{r'}$ encounters $r$ when it returns to its source. At round $(d+1)T+1$, $r$ records the information of $F_{r'}$ in a temporary data structure $F'_r$, and the DFS continues.  

In the next interval, from $(d'T+2)$ to $(d'T+2)+(12c-1)$ with $d'=d+1$, $r$ traverses the nodes of $F_{r'}$, exits through an inter-group edge, and proceeds to the child group $F_{r''}$. After waiting until $(d'+1)T$, it updates its record with the structural information of $F_{r''}$ and resumes traversal.  

     \begin{algorithm}[ht!]
    \caption{$Oscillating(r)$}
    \label{algo:osci}

    \lIf{$round==dT+1$  \label{alg:osci1}}
    {
        wait at the current node  \label{alg:osci2}
    }

    \If{$round=((dT+1),(dT+1)+12c]$ \label{alg:osci3}}
    {
        \If{$recent=false$ \label{alg:osci4}}
        {
            oscillate within its group \label{alg:osci5} \\
            \If{$\exists$ \activee~$r_a$ or \lead~$r_l \in R[v]$ \label{alg:osci6}}
            {
                wait at the current node \label{alg:osci7}
            }
        }
        \lElse
        {
            wait at the current node \label{alg:osci9}
        }
    }
    \If{$round==(dT+1)+12c+1$ \label{alg:osci10}}
    {
        \If{$r.recent=false$, $r.move=0$ \label{alg:osci11}}
        {
            \If{$r.\sourcenode(v)=\yes$ in $r.F$, $r.count=\delta_v-1$ with $v\in Q_r$ \label{alg:osci12}}
            {
                $r.count=r.count+1$, $r.DFS=\done$ \label{alg:osci13}
            }
            
            \Else
            {
                \If{$r.\sourcenode=\yes$ in $r.F$, $r.count\neq \delta_v-1$ with $v\in Q_r$ \label{alg:osci15}}
                {
                    $r.count=r.count+1$ \label{alg:osci16}
                }
            \If{$\exists$  \lead~robot $r_l \in R[v]$ \label{alg:osci17}}
            {
                shares the necessary information \label{alg:osci18}\\
                \If{$r.\groupsize<2c$ \label{alg:osci19}}
                {
                    $r.state=\lead$ and $\portext=p'$ \label{alg:osci21} \\
                }
                \Else
                {
                    \If{$v \in F_r$ \label{alg:osci23}}
                    {
                        \If{$p'\neq parent(v)$ \label{alg:osci24}}
                        {
                            $\portext= p'$ and $r.recent=true$ \label{alg:osci25}
                        }
                    }
                    \lElse
                    {
                        $r.move=1$, $r.\portext=parent(v)$ \label{alg:osci27}
                    }
                }
            }
            \Else
            {
                \If{$p' \neq parent(v)$ \label{alg:osci29}}
                {
                   $\portext= p'$ and $r.recent=true$ \label{alg:osci30}
                }
            }
        }
    }
   
    \ElseIf{$r.recent=true$ \label{alg:osci31}}
    {
        \If{$\nexists$ \osci~robot $r_o \in R[v]$ with $v \in Q_{r_0}$ \label{alg:osci32}} 
        {
            $r.gr\_Child[v'] = r.portext(v')$ in $r.F$, where $v'$ is previous traversed node \label{alg:osci33}    
        }
        $r.move=1$, $r.recent=false$ and
       $\portext=p$ \label{alg:osci34}
    }
    \ElseIf{$r.recent=false$ and $r.move=1$ \label{alg:osci35}}
    {
        $r.move=0$ and exit through \portext \label{alg:osci36}
    }
    }
    \If{$round=((dT+1)+12c+1, (d+1)T+1)$ \label{alg:osci37}}
    {
        \lIf{$recent=false$ \label{alg:osci38}}
        {
            oscillate and return to $gr\_source$ node \label{alg:osci39}
        }
        \lElse{wait at the current node \label{alg:osci41}}
        
    }
    \If{$round==(d+1)T+1$\label{alg:osci42}}
    {
        \lIf{$recent=true$ \label{alg:osci43}}
        {
            exit through \portext \label{alg:osci44}
        }
        \lElseIf{$r.DFS=\done$ \label{alg:osci45}}
        {
            run Phase 2 \label{alg:osci46}
        }
    }
\end{algorithm}


If $r$ reaches a leaf group $F_{r^l}$, it backtracks to the parent group together with the guard $r^l$, updating $r^l.gather= \yes$ (initially set to \no). While backtracking, $r$ systematically explores any unexplored child groups. Once all child groups of a node are visited, the traversal continues upward, carrying \osci~robots that have their $gather$ flag set to \yes. This DFS-based process guarantees that information from every group is collected and propagated. Ultimately, the $osci\_head$ accumulates the complete structural information of $G$.

\subsubsection{Phase 3: Memory-efficient dispersion}
After Phase 2, the robot $r$ perform a local computation at $S$, setting the variables $dsg\_gr$, $d$, $stmp\_gr$, $tmp\_gr$, $ctr$, $level$, and data structure $D$ (Section~\ref{ques:phas3}). The robots now follow a two-step cooperative strategy.

\begin{description}[leftmargin=0em]

\item[Stage 1:] The robots disperse according to their assigned groups, stored in the variables $tmp\_gr$ and $stmp\_gr$, following a DFS, as described in Section~\ref{des:gathering} for the $osci\_head$ robot. Robots having non-\textsc{Null} values of either $tmp\_gr$ or $stmp\_gr$ initiate DFS and traverse the groups by collaborative guidance. Each robot maintains a counter variable $round$ (initialized to~0) incremented by 1 in every round elapses. Each robot spends $12c$ rounds at each group.

During backward traversal from supernode $F^i$, where $F^i$ represents a group with group number $i$, once all child groups are visited, the robot $r$ with $r.stmp\_gr = F^i$ stays at the node, while the remaining robots backtrack and continue traversal. 



\item[Stage 2:] Robots initiate Stage~2 when $round = 12c \cdot k$ and subsequently reset $round = 0$, marking the completion of Stage~1. A robot $r$, with $r.level$ is $l_i$ performs its computation when $round = 12c \cdot l_{i-1}$. The robot $r$ in group $F^j$, where $r'.stmp\_gr = F^j$, executes computations based on data stored in $r'.D$. Specifically, $r$ either moves to a child group $F^{j'}$ to fulfill the pending demands, guided by the robot $r'$, or remains at its current group $F^i$.  

After $12c \cdot l_d$ rounds, where $l_d$ is the diameter of $\mathcal{T}$, each node hosts at most one robot of the same color and become idle.


    
    


\end{description}

\subsection{Analysis of the Algorithm}

\begin{lemma}\label{lem:1grexp}
Let $F_r$ be the group formed by the robot $r$. Let $Q_r$ and $E_{Q_r}$ be the sets of \required~nodes and \required~edges of $F_r$, respectively, and let $H_r$ and $E_{H_r}$ be the sets of \helper~nodes and \helper~edges of $F_r$, respectively. Define $U_r = Q_r \cup H_r$ as the set of all nodes in $F_r$, and let $E_r = E_{Q_r} \cup E_{H_r}$ be the set of all edges in $F_r$, where for any two vertices $v_i,v_j \in U_r$, $E_{Q_r} = \{(v_i,v_j) \mid v_i \in Q_r \text{ or } v_j \in Q_r\}, \text{and } E_{H_r} = \{(v_i,v_j) \mid v_i, v_j \in H_r\}$. Then the following properties hold for the group $F_r$:

\begin{enumerate}[noitemsep, topsep=0pt,leftmargin=1.2em]
    \item $(U_r, E_r)$ forms a tree rooted at the group source node $v \in Q_r$, i.e., $r.\grsource(v) = \yes$.
    
    
    
    \item The number of nodes in \Fr~satisfies $|U_r|\leq 6c$.
    
    
    \item Let $v$ be the group source node of $U_r$ (i.e., $r.\grsource(v)=\yes$), and let $u \in Q_r$ be the farthest \required~node from $v$. Starting from $v$, robot $r$ visits every node of $U_r$ at least once and reaches $u$ within at most $12c$ rounds.
         
\end{enumerate}
\end{lemma}

\begin{proof}~
\begin{enumerate}
    \item We show that $(U_r,E_r)$ forms a tree by proving that it is acyclic and connected. Let robot $r$ become the leader at node $ v$. By the algorithm, $r$ initializes its group $F_r$ by including $v$ as the source node.

\textbf{Acyclicity.} A vertex is added to $U_r$ exactly once: either upon its first visit during the \forward~traversal (as a \required~node) or during the \backtrack~phase when robot $r$ expands its group and incorporates vertices from another group as \helper~or \required~nodes. During the execution, robot $r$ may expand or shrink its group.

During the forward traversal, if robot $r$ encounters a vertex $u \in U_r$ that has already been visited (i.e., $u \in Q_r$ or $u \in H_r$), the \activee~robot detects this condition and, after the prescribed waiting time, backtracks without adding any edge. Hence, no edge is added between two previously visited vertices in the forward phase.

During the backtracking phase, robot $r$ may expand its group upon reaching a vertex $v \in Q_{r'}$ for some robot $r'$. The expansion follows a simple path, and for each newly added vertex, exactly one edge, connecting it to its predecessor on the path is added to $E_r$. Suppose, for contradiction, that this expansion creates a cycle. Then some newly added vertex $v'$ must be connected to a vertex $u \in U_r$ that was already present before the expansion. This would imply that the robot $r$ traversed the edge from $v'$ to $u$, establishing a parent-child relationship with $v'$ as the parent. However, during the backtracking phase, traversal is allowed only from a child to its parent; thus, the robot $r$ would have backtracked from $u$ to $v'$, a contradiction. Therefore, no cycle is formed during expansion.

Finally, when the robot $r$ shrinks its group, vertices and their incident edges are removed. Since edge deletions cannot create cycles, acyclicity is preserved.

Thus, in all cases, the subgraph $(U_r,E_r)$ maintained by robot $r$ is acyclic.

    \textbf{Connectivity.} Let $u, v \in U_r$ be two vertices in $F_r$. If the path between $u$ and $v$ consists solely of vertices belonging to $U_r$, then there exists a path between $u$ and $v$ using only edges in $E_r$, as whenever $r$ discovers a new vertex $v \in U_r$ from an already visited vertex $u \in U_r$, exactly one edge $(u,v)$ is added to $E_r$. Hence, $u$ and $v$ are connected in $F_r$.

    Suppose, for contradiction, that there exists a vertex $w \notin F_r$ lying on the path between $u$ and $v$. According to the algorithm, if the robot $r$ expands its group from $u$ during the forward phase and encounters a vertex $w \notin U_r$ that has already been visited, then $r$ immediately backtracks. This implies that the robot cannot proceed toward $v$, contradicting the assumption that $v \in U_r$.

    Alternatively, if $r$ expands its group during the \backtrack~phase and encounters such a vertex $w$, then the algorithm ensures that the vertices on the path from $w$ to $v$ are added as \helper~or \required~node. Consequently, there exists a path from $u$ to $v$ entirely within $F_r$.

    Finally, when robot $r$ shrinks its group, connectivity is preserved, as after the shrinking operation, either the vertex $v$ continues to belong to $F_r$, or it is removed along with its dependent subtree, ensuring that the remaining vertices in $F_r$ remain connected.
    
    Moreover, since the exploration maintains acyclicity, this path is unique. Therefore, the subgraph induced by $F_r$ is connected.


    Since $(U_r, E_r)$ is both acyclic and connected, it follows that it forms a tree rooted at the source node $v$.

    \item In a group $F_r$, the number of \required~nodes satisfies $2c \leq |Q_r| \leq 4c$. Therefore, $|Q_r| \leq 4c$, and it remains to show that the number of \helper~nodes satisfies $|H_r| \leq 2c$.
    
    By the algorithm, the \helper~nodes are introduced in $F_r$ (if needed) when a robot $r$ encounters a node that has already been visited and belongs to the group of another robot. Suppose robot $r$, with group size $g < 2c$, reaches a visited node $v$. If $g > 2c$, then $r$ has sufficient size and simply returns to its group without adding $v$ as a \helper~or a \required~node. Let us assume that the robot $r$ is backtracking to node $v$. If $r$ visits $v$ (where $v\notin Q_r$ or $v \in Q_r$) in the forward phase, it will do nothing and exit $v$ through the port it entered, independent of its size, as specified by the algorithm. Let $r_o$ be the \osci~robot that stores the information corresponding to $v$ and is present at $v$.

When $r$ reaches $v$, there are two possibilities:

\begin{itemize}
    \item[(i)] If $v$ has at least one unvisited child $v_1$, then $r$ proceeds in the \forward~phase to $v_1$. In this case, $v$ is included in $H_r$ as a \helper~node. At $v_1$, $r$ stores the information of $v_1$, and the group size increases to $g+1$. If $v_1$ has no unvisited children, then $r$ eventually backtracks to $v$ (as $r$ will be visiting in \forward~phase) according to the graph exploration algorithm. Once all children of $v$ have been explored, $r$ augments $F_r$ with the information of the subtree rooted at $v$. 

    \item[(ii)] If all children of $v$ have already been visited when $r$ arrives, then $r$ immediately receives the subtree information from $r_o$ and integrates it into $F_r$.
\end{itemize}

In both cases, the critical point is that the maximum amount of information that $r$ can receive from the subtree of $v$ (via $r_o$) is bounded by $4c$ nodes. The worst-case situation for the robot $r$ arises when the \required~nodes are farther away from $v$ (i.e., up to $4c$ hops), so that $r$ must traverse a larger number of intermediate nodes (\helper~nodes) in order to reach the corresponding nodes whose information is shared by $r_o$. Let the size of the subtree of $v$ corresponding to $r_o$, including $v$, be $4c$ (maximize size). Thus, the nodes relevant to $r$ may lie at most at a distance $4c$ from $v$, and in this extreme case, $r$ is forced to include more \helper~nodes than \required~nodes, since the path to each \required~node may pass through additional intermediate nodes. 

Now, if the group size is $g < 2c$, the group formation procedure continues. In particular, when $g = 2c-1$, robot $r$ attempts to store the minimum additional \required~nodes from the subtree of $v$ so that the \required~nodes are farther away from $v$. In this case, the robot $r_o$ contributes $2c+1$ nodes of information. Consequently, $r$ stores exactly $2c+1$ nodes, while the remaining $2c-1$ nodes serve as \helper~nodes that guide $r$ to reach the \required~nodes. Hence, $F_r$ contains at most $2c-1$ \helper~nodes.

 Hence proves that $|U_r|=|Q_r|+|H_r|\leq 6c$.

 \item Let $v$ be the group source node of $U_r$ , and let $u \in Q_r$ be the farthest \required~node from $v$. From part (ii), the group $F_r$ contains at most $4c$ \required~nodes and at most $2c-1$ \helper~nodes. Hence, the total number of nodes in $F_r$ is strictly less than $6c$, and $F_r$ forms a tree.

During the DFS traversal, the robot $r$ may visit each node multiple times due to the forward and backtrack phases. Specifically, each \required~node is visited at least once (forward phase) and at most twice (forward and possibly backtrack), and each \helper~node is visited at least once (when serving as a connector) and at most twice (forward and backtrack).

Therefore, in the worst case, traversing all \required~nodes requires at most $2 \cdot 4c = 8c$ rounds, and traversing all \helper~nodes requires at most $2 \cdot (2c-1) < 4c$ rounds. Summing these bounds, the total number of rounds required to visit all nodes in $F_r$ is at most $8c + 4c = 12c$.

Hence, starting from $v$, robot $r$ visits every node of $U_r$ at least once and reaches the farthest \required~node $u$ within at most $12c$ rounds. \qedhere
\end{enumerate}
\end{proof}

\begin{lemma}\label{lem:2finishedexpr}
Let $F_r$ be the group formed by robot $r$ (in either the \lead~or \osci~state), and let $S \in F_r$ be the source node of the given graph $G$. The graph exploration process ends when the variable $count$ maintained by $r$ satisfies $r.count = \delta_S$, where $\delta_S$ is the degree of the node $S$.

\end{lemma}

\begin{proof}
    
Let $F_r$ denote the group formed by robot $r$, and let $S \in F_r$ be the source node of the graph $G$. The robot maintains a variable $count$ to record the number of incident edges of $S$ that have been completely explored.

According to the exploration algorithm, the robot $r$ increments $count$ whenever the \activee~robot leaves the source node $S$ through a port $p$ and later returns to $S$ through the same port during the backtracking phase, indicating that the corresponding incident edge has been completely explored. Since the degree of $S$ is $\delta_S$, there are exactly $\delta_S$ such incident edges.

Therefore, the exploration terminates precisely when $count = \delta_S$, which guarantees that all edges incident to $S$ have been explored and that robot $r$ has returned to the source node, thereby completing the exploration of the graph.
\end{proof}

As defined earlier, $F_r$ denotes the group formed by robot $r$. For ease of analysis, we classify groups according to their group number (\groupno). Let $W_1, W_2, \ldots, W_b$ be the collection of groups such that $W_i$ corresponds to the group of the robot $r$ with $r.gr\_num = i$. For convenience, we also refer to each group as a \emph{supernode}. Thus, $W_i$ represents the supernode (or equivalently group) $F_r$, where $r$ is the robot located at the source node $v$ of $F_r$ (i.e., $r.\grsource(v) = \yes$ for some $v \in F_r$) and $r.\groupno = i$. Similarly, we denote by $Q_i$ and $H_i$ the sets of \required~and \helper~nodes, respectively, of the supernode $W_i$. Let ${\cal W} = \{W_1, W_2, \ldots, W_b\}$ be the collection of all supernodes. Clearly, $W_1$ is the first group created during the traversal, which means $S \in Q_1 \subseteq W_1$ or equivalently $S\in \Fr$ for some robot $r$ such that $r.\sourcenode(S) = \yes$. From now on, we use the term supernode in place of group.  We now prove some properties with respect to the groups.

\sloppy

\begin{lemma}\label{lem:3endexpr}
At the end of the graph exploration process, the following properties hold for the set of groups/supernodes:
\begin{enumerate}[noitemsep, topsep=0pt,leftmargin=1.2em]
    \item The source node $S\in W_1$ corresponding to a robot $r$ (in \lead~or \osci~state) contains at least two robots. 
    \item For every node $v \in V$ of the graph $G$, there exists a unique supernode $W_i$ such that $v$ is a \required~node of $W_i$. Equivalently, the sets of required nodes form a partition of $V$, i.e., $Q_i \cap Q_j = \emptyset$ for all $i \neq j$, and $Q_1 \cup Q_2 \cup \cdots \cup Q_b = V$, where $Q_i$ denotes the set of \required~nodes of supernode $W_i$.

    \item The number of supernodes satisfies $|\mathcal{W}| \leq k/2$, $k$ is the total number of robots. (Equivalently, $b$ satisfies $b \leq k/2$.)

    
    \item Let $\mathcal{E}$ be the set of inter-group edges identified within the set $\mathcal{W}$ by our algorithm. Then the structure $\mathcal{T} = (\mathcal{W}, \mathcal{E})$ forms an oriented tree rooted at $W_1$, where $W_1$ is the supernode containing robot $r$ with $r.\sourcenode(S) = \yes$ for $S \in W_1$. \label{lem:5_suprtree}

\end{enumerate}
\end{lemma}

\begin{proof}~

    \begin{enumerate}
    
        \item We will prove that, at the end of the exploration process, the source node contains at least two robots. Consider a graph $G$ consisting of $n$ nodes. During the exploration process, each settled robot is responsible for storing the information of at least $2c$ nodes and at most $4c$ nodes. Let us assume that all robots settle during the exploration, forming $k$ distinct groups. Each group then contains between $2c$ and $4c$ nodes.  

    Thus, the total number of nodes covered in this case satisfies $k \cdot 2c \;\leq\; n \;\leq\; k \cdot 4c$.

    Since the number of robots initially is $k = \lceil n/c \rceil$, it follows that $k \cdot 2c > n$.
    
    This inequality leads to a contradiction, as it implies that not all robots can settle during the exploration process. Therefore,  some robots must return to the source node $s$.  

    Moreover, since at least one robot is designated to store the information of the source node $s$ itself (in \osci~state), and at least one other robot (in \lead~state) necessarily returns to $s$ after completing the exploration, the source node $s$ will contain at least two robots at the end of the process. This proves the lemma.

    \item We prove this by contradiction. Suppose that there exists a node $v$ such that $v \notin Q_i$ for every robot $r$.  

By the algorithm, a node $v$ is included in $Q_i$ when robot $r$ visits $v$ for the first time during the exploration process. At the end of the exploration process, every node in the graph is eventually visited by some robot. If $v$ does not belong to any $Q_i$, it means that no robot has ever visited $v$ during the entire exploration. However, this contradicts the exploration algorithm, which guarantees that the traversal covers the entire graph and no node remains unvisited.  

Hence, our assumption is wrong, and it follows that for every node $v$, there exists at least one robot $r$ such that $v \in Q_i$ and $\bigcup_{i=1}^b Q_i = V$.

Now, we have to prove that $v \in Q_i$ is unique. According to the algorithm, an unvisited vertex $v$ is added to a group $W_i$ (as a \required~node) at the time it is first discovered. Concretely, once a vertex $v$ is placed into $Q_i$ the visited flag at $v$ prevents any subsequent group from claiming $v$ as a \required~node.

Consider a robot belonging to another group $W_j$, where $j \neq i$, that later reaches $v$. Upon reading that $v$ is visited, the robot follows the algorithm and either uses $v$ as a \helper~node or backtracks, depending on the traversal phase. If the group $W_j$ expands at $v$, then the group $W_i$ is compressing, ensuring that a single group preserves $v$ as a \required~node. Hence no vertex can belong to two distinct $Q_i$ and $Q_i \cap Q_j = \varnothing$ for all $i\neq j$.

Combining the above cases, we obtain that $\{Q_1,\dots,Q_b\}$ is a partition of $V$, as required.

\item During the exploration each settled group (supernode) $W_i$ is responsible for the information of between $2c$ and $4c$ required nodes. By part (iii) the sets $Q_1,\dots,Q_b$ form a partition of $V$, so every vertex of $G$ is assigned to exactly one group. Therefore, the total number of robots required to cover $n$ nodes is at least $k/2$, where each group $W_i$ covers $2c \le |Q_i| \le 4c$ nodes.

\item Without loss of generality let $W_1$ be the supernode that contains the global source node $S$ and the robot $r$ with $r.\sourcenode(S)=\yes$. We direct every inter-group edge from a parent toward the child group and child to parent,; thus $\mathcal T$ is an oriented graph. We prove that $\mathcal T$ is a rooted tree with root $W_1$ by establishing acyclicity and connectedness.

\textbf{Acyclicity.} Assume, for the sake of contradiction, that $\mathcal T$ contains a directed cycle $W_{i_1} \to W_{i_2} \to \cdots \to W_{i_k} \to W_{i_1}$ with $k \ge 2$. By the construction of the algorithm, an inter-group edge $(W_p, W_q) \in \mathcal E$ is created if and only if the group $W_q$ is newly formed by robots originating from the group $W_p$. Therefore, $W_p$ must exist before $W_q$ is created.

However, the existence of the edge $W_{i_k} \to W_{i_1}$ implies that $W_{i_1}$ is created after $W_{i_k}$, which contradicts the fact that $W_{i_1}$ is an ancestor of $W_{i_k}$ in the directed cycle. Hence, such a directed cycle cannot exist.

Moreover, during both the expansion and shrinking operations of a group $W_{i_\ell}$, the algorithm preserves the parent--child relationship among groups and never introduces an edge to a previously created ancestor group. Therefore, the acyclicity property is maintained throughout the execution.

Thus, $\mathcal T$ is acyclic.

\textbf{Connected.}
The exploration process starts with the source group $W_1$. Whenever a new group $W_j$ is formed, it is created from an existing group $W_i$ that has completed its group-formation procedure. During this transition, the robots explicitly record an inter-group edge connecting $W_i$ to $W_j$. Consequently, each newly formed group has a well-defined incident edge linking it to its predecessor.

By repeatedly following these recorded inter-group edges, any group $W_j$ can be reached from the source group $W_1$.

If a group $W_i$ subsequently shrinks, then according to the algorithm, it updates its inter-group edges so that connectivity with the expanding group $W_j$ is preserved. No group is disconnected during this process.

Therefore, every group remains connected to the source group, and the structure $\mathcal T$ is connected.

Since the underlying undirected graph of $\mathcal T$ is connected and the directed parent relation is acyclic with a unique root $W_1$, $\mathcal T$ is an oriented tree rooted at $W_1$.\qedhere
\end{enumerate} \end{proof}

\begin{lemma} \label{lem:5uniquedes}
Consider the tree ${\cal T} = (\mathcal{W}, \mathcal{E})$, where all robots, having collectively acquired the connectivity information of ${\cal T}$ (of $G$), are located at the source node $S\in W_1$. After local computation, each robot $r$ is assigned a unique destination supernode $W_p$, that is, $r.dsg\_gr = W_p$, such that there exists a node $v \in W_p$ satisfying $\mathsf{col}(r) = \mathsf{col}(v)$. 
\end{lemma}
\begin{proof}
We prove the lemma in two parts: existence and uniqueness of the robot-to-supernode assignment.

By construction, each supernode $W_i$ contains a set of required nodes, and for every robot color, there exists at least one required node in some supernode whose color matches the robot. Let $r$ be an arbitrary robot. Since all robots have complete knowledge of the supernodes and their required nodes, there exists at least one supernode $W_p$ that contains a required node $v$ with $\mathsf{col}(v) = \mathsf{col}(r)$. Therefore, an assignment of $r$ to $W_p$ satisfying the color constraint exists.

Assume for contradiction that a robot $r$ is assigned to two supernodes $W_p$ and $W_q$, or that two robots $r_1$ and $r_2$ are assigned to the same required node $v$. By the assignment rule, a required node is marked as assigned at the moment a robot is allocated to it. Let $v \in W_p$ be the node assigned to $r_1$. When $r_2$ attempts the assignment, it sees that $v$ is already assigned and therefore cannot select it. Similarly, a robot cannot be assigned to two supernodes simultaneously, because assignment is performed atomically based on the first unassigned node matching the robot’s color. This contradicts the assumption. Hence, each robot is assigned exactly one supernode, and each required node is assigned to at most one robot.

Combining existence and uniqueness, we conclude that after local computation, each robot $r$ has a unique destination supernode $W_p$ such that there exists a node $v \in W_p$ with $\mathsf{col}(r) = \mathsf{col}(v)$. 
\end{proof}

\begin{lemma}\label{lem:6disperse}
Let $\mathcal{W_s} = {W_{i_1}, \ldots, W_{i_l}}\subseteq \mathcal{W}$ be the set of destination supernodes for which some robot $r$ satisfies $r.dsg\_gr = W_{i_j}$. After Stage 1 of Phase 3, the following hold:
\begin{enumerate}[noitemsep, topsep=0pt,leftmargin=1.2em]
        
    \item Let $r\in\mathcal{R}$ be any robot with $r.dsg\_gr = W_{i_j}$ for some $W_{i_j}\in\mathcal{W_s}$. Let $\pi(r)=(W_{q_1},W_{q_2},\dots,W_{q_t})$, where $W_{q_\ell}\in \mathcal{W}$ be the sequence of supernodes on the downward path from $W_{q_1}$ to $W_{i_j}$. Then, (a) if  $r.stmp\_gr = W_{q_1} \in \mathcal{W}-\{W_{i_j}\}$ or $r.stmp\_gr=\textsc{NULL}$ i.e., $r$ is at source $S$ ($W_{q_1}=W_1$), then there exists a path $\pi(r)$ from the supernode $W_{q_1}$ to $W_{i_j}$, where each supernode $W_{q_\ell}$ along the path contains at least one robot, allowing $r$ to reach its destination supernode $W_{i_j}$ efficiently. (b) If $r.stmp\_gr = W_{i_j}$, then $r$ is at destination supernode.

    
    \item Each supernode $W_{i_j} \in \mathcal{W_s}$ contains exactly one robot.

    \item For every robot $r \in W_{i_j}$ with $r.stmp\_gr = W_{i_j}$, the counter variable $ctr$ and the data structure $D$ is maintained, which stores the number of robots required in the subtree rooted at $W_{i_j}$ and the colors of robots needed for each child supernode, along with the corresponding \groupno~requesting additional robots.
\end{enumerate}
\end{lemma}

\begin{proof}~

\begin{enumerate}

    \item Let $r \in \mathcal{R}$ be a robot with $r.dsg\_gr = W_{i_j} \in \mathcal{W_s}$. 

    \textbf{Case (a):} By Lemma~\ref{lem:5_suprtree}, $\mathcal{T}$ is a tree. Hence, there exists a unique path $\pi(r)$ in $\mathcal{T}$ from the supernode $W_{q_1}$ to $W_{i_j}$. We show that every supernode along the path $\pi(r)$ contains at least one robot.

    We prove this by contradiction. Suppose that $r.stmp\_gr = W_{q_1} \neq W_{i_j}$, or $r.stmp\_gr = \textsc{NULL}$ (i.e., $r$ is located at the source supernode $W_1$), and assume that there exists a supernode $W_{q_p}$ on the path from $W_{q_1}$ to $W_{i_j}$ that does not contain any robot.

    By the construction of the algorithm, a robot with a destination group $r.dsg\_gr = W_{i_j}$ is placed either at $W_{i_j}$ itself or at the nearest ancestor supernode of $W_{i_j}$ in the tree $\mathcal{T}$. Since $r.stmp\_gr = W_{q_1} \neq W_{i_j}$, robot $r$ cannot be placed at $W_{i_j}$. 

    However, if $W_{q_p}$ is empty, then $r$ should have updated its temporary group to $r.stmp\_gr = W_{q_p}$, contradicting the assumption that $W_{q_p}$ is empty. Therefore, every supernode along the path $\pi(r)$ contains at least one robot.

    \textbf{Case (b):} Suppose $r.stmp\_gr = W_{i_j}$. By definition, $r$ is already located at its destination supernode, and no further traversal is required.  

    In both cases, either $r$ is already at its destination or a path exists from its current location to $W_{i_j}$ through supernodes containing robots. Hence, every robot assigned a destination supernode has a guaranteed route to reach it, which proves the statement.


    \item Consider a supernode $W_{i_j} \in \mathcal{W_s}$. By the assignment procedure, exactly one robot $r$ is designated to have $r.stemp\_gr/ttemp\_gr = W_{i_j}$. The algorithm ensures that no other robot is assigned the same destination supernode, as each assignment is unique and conflicts are resolved by the minimum-ID rule.  
    Moreover, any robot that might traverse $W_{i_j}$ as an intermediate node (i.e., for traversal purposes) does not become a resident of $W_{i_j}$, because only robots whose destination is $W_{i_j}$ remain there. Therefore, at the end of Stage 1, exactly one robot occupies $W_{i_j}$.  
    Thus, each supernode $W_{i_j} \in \mathcal{W_s}$ contains exactly one robot.

    \item Let $r \in W_{i_j}$ be a robot with $r.stmp\_gr = W_{i_j}$. By construction, $r$ is the robot assigned to manage the destination supernode $W_{i_j}$. During the local computation of Phase 3, $r$ collects information from all child supernodes $W_{\eta_1}, \dots, W_{\eta_\delta}$ of $W_{i_j}$ regarding the number and colors of robots required to satisfy the demands of their respective subtrees.  
    
    The algorithm explicitly updates the counter variable $r.ctr$ to reflect the total number of robots required in the subtree rooted at $W_{i_j}$, including the demand of $W_{i_j}$ itself. Simultaneously, $r$ constructs and maintains a data structure $D$ that records, for each child supernode, the number of robots needed, the specific colors required, and the group numbers requesting additional robots.  

 Therefore, the counter variable $ctr$ and data structure $D$ are correctly maintained for every robot $r$ with $r.stmp\_gr = W_{i_j}$. \qedhere
\end{enumerate}
\end{proof}

\begin{lemma}\label{lem:7S2P3}
 For each $W_{i_j} \in \mathcal{W_s}$, let $d$ denote the number of robots whose destination supernode is $W_{i_j}$. At the end of Stage 2 of Phase 3, exactly $d$ robots are placed in $W_{i_j}$, and each node $v \in W_{i_j}$ hosts at most one robot $r$ with $\mathsf{col}(v) = \mathsf{col}(r)$.

\end{lemma}
\begin{proof}

Let $W_{i_j} \in \mathcal{W_s}$ be a destination supernode, and let $d$ denote the number of robots with $r.dsg\_gr = W_{i_j}$. By construction, Stage 1 ensures that each robot assigned to $W_{i_j}$ holds information about its destination supernode, as well as the path from the source to $W_{i_j}$ via the $stmp\_gr$ and $tmp\_gr$ variables.  

During Stage 2, robots move along these paths according to the data structures $D$ maintained at each supernode. The DFS-based placement procedure guarantees that robots are routed toward their respective destination supernodes without conflicts. Since each robot is uniquely assigned to a supernode and respects the color of the nodes it occupies, no two robots of the same color are placed on the same node.  

Consequently, at the end of Stage 2:  

\begin{enumerate}
    \item Exactly $d$ robots are placed in $W_{i_j}$, as each robot assigned to $W_{i_j}$ reaches its destination without duplication or loss.
    \item Each node $v \in W_{i_j}$ hosts at most one robot $r$ such that $\mathsf{col}(v) = \mathsf{col}(r)$, ensuring that color-matching constraints are preserved.
\end{enumerate}  

Hence, the placement of robots in $W_{i_j}$ is correct, and the lemma follows.
\end{proof}

\begin{theorem}\label{thm:rooted-known}
\cd~can be solved deterministically in rooted initial configuration with $k>1$ robots having prior knowledge of $n$ in $O(n/k \cdot m)$ rounds using $O(n/k \cdot \log(k+\Delta))$ bits per robot.
\end{theorem}
\begin{proof}
The theorem follows from Lemmas~\ref{lem:1grexp}-\ref{lem:7S2P3}. Consequently, $\cd$ is solved deterministically in a rooted initial configuration in $O(n/k\cdot m)$ rounds. Each robot $r$ stores only the information associated with its own group $W_i$, which requires $O(n/k\cdot \log(k+\Delta))$ bits of memory. Hence, $\cd$ can be solved deterministically within the stated time and memory bounds.

\end{proof}

\section{Arbitrary Graph: Rooted, \texorpdfstring{$k\leq n$}{k<n}, Unknown  \texorpdfstring{$n$}{n} \label{sec:rooted-unknown}}


In this section, we study the \cd~problem under a rooted configuration, where all $k$ robots are initially colocated at a single source node $S$. The total number of nodes $n$ in the graph is unknown. As in the rooted setting, all robots are initially colocated at a single node, enabling each of them to determine the total number of robots $k$. Thus, whether $k$ is initially known or unknown does not affect the model in the rooted configuration. 

The core idea is to enable robots to explore an arbitrary graph without prior knowledge of $n$ by \emph{iterative guessing} with \emph{DFS-based exploration} (as in Section~\ref{sec:rooted-known}). Since the known-$n$ strategy of Section~\ref{sec:rooted-known} relies on bounding DFS by $n$, robots begin with the assumption $n=k$ and initiate a DFS, where each node is assigned to at most one robot and groups are coordinated by the  \osci~robot at each group. If the traversal completes and all robots are placed without leaving nodes unexplored, the guess is correct and dispersion succeeds. Otherwise, the inability to settle all robots (or the discovery of unexplored nodes) reveals that the guess was too small. In this case, all robots backtrack and regroup at $S$, the assumed value of $n$ is doubled, and a new DFS attempt is initiated. By progressively doubling the guess, the robots eventually reach an interval $[2^i k, 2^{i+1}k]$ that contains the true $n$, ensuring that dispersion is completed in that iteration. This strategy guarantees correctness regardless of the initial underestimate of $n$, while incurring only a logarithmic overhead of $O(\log(n/k))$ in number of DFS trials compared to the known-$n$ setting. Once the graph exploration phase is completed, the gathering and dispersion steps from Section~\ref{sec:rooted-known} are applied, thereby completing \cd.

\begin{theorem}\label{thm:rooted-unknown}
\cd~can be solved deterministically in rooted initial configuration with $k>1$ robots having no prior knowledge of $n$ 
in $O(\log (n/k) \cdot n/k \cdot m)$ rounds using $O(n/k \cdot \log(k+\Delta))$ bits per robot. 
\end{theorem}

\begin{proof}
The algorithm follows an iterative guessing strategy combined with DFS-based exploration.

Initially, the robots assume that $n = k$ and perform a bounded DFS starting from the source node $S$. During this traversal, one robot is assigned to each group, ensuring that every group contains at least $2c$ and at most $4c$ nodes. If the traversal completes successfully, that is, if $count = \delta_S-1$ and some robots return to the source node $S$, then dispersion is achieved by Lemma \ref{lem:2finishedexpr}.

Otherwise, suppose that during the exploration only one \activee~robot remains and it detects the existence of unexplored nodes. This implies that the traversal is incomplete. In this situation, the remaining \activee~robot becomes a \lead, and no \activee~robot returns to $S$, which contradicts Lemma \ref{lem:2finishedexpr}. Hence, the current estimate underestimates the actual value of $n$. In this case, all robots return to the source node $S$ by backtracking along the DFS traversal, double the guessed value of $n$, and repeat the DFS process.

Since the guessed value of $n$ doubles in each iteration, after $O(\log (n/k))$ iterations, the estimate becomes sufficiently accurate, and the DFS completes successfully. Each DFS attempt requires $O(n/k \cdot m)$ rounds, as established in the known-$n$ case (Theorem \ref{thm:rooted-known}). Therefore, the total round complexity of the algorithm is $O(\log(n/k) \cdot n/k \cdot m)$.

The memory requirement per robot remains the same as in the known-$n$ setting, since robots store only local DFS state and coordination information. Consequently, each robot uses $O(n/k \cdot \log (k+\Delta))$ bits of memory.

Hence, the dispersion problem in a rooted configuration with unknown $n$ can be solved deterministically within the stated time and memory bounds.
\end{proof}

\section{Arbitrary Graph: General, \texorpdfstring{$1<k\leq n$}{1<k<n} \label{sec:general-known}}

We study \cd~for general initial configurations, that is, there exists a source node $s$ in the graph $G$ that initially hosts at least two robots. In any general initial configuration, a node hosting exactly one robot is called a $single$ source, while a node containing at least two robots is a $multi$ source. Our approach builds on the multi-source DFS framework of Kshemkalyani and Sharma~\cite{kshemkalyani2025near}, which is applied to classical dispersion and adapted to preserve the group structure required for \cd.

\textbf{High level idea:} A robot located at a $single$ source remains at its current node, since a single robot without node memory cannot explore the graph, as discussed in Section \ref{sec:lower-bound-paragraph}. In contrast, robots at a $multi$-source initiate a DFS traversal, forming bounded-size groups coordinated by a designated \osci~robot (as in Section~\ref{sec:lower-bound-paragraph}). When two DFS traversals meet at a node or encounter a $single$-source robot, a subsumption step is executed. The traversal with more settled robots survives, while the smaller traversal is collapsed and its robots are reassigned to the surviving DFS; ties are broken by the minimum {\it treelabel}. Multiple simultaneous collisions are resolved via pairwise subsumptions. Over time, the number of concurrent DFS traversals decreases, ensuring that eventually a single traversal remains. 
 
\subsection{Special Case: When \texorpdfstring{$k= n$}{k=n} }  

When $k=n$ and both $k$ and $n$ are known, each source starts a DFS forming groups of size $[2c,4c]$, and subsumption ensures that the dominant traversal continues exploration.

\textbf{Algorithm:} Each robot initializes a {\it treelabel} as the minimum-ID robot at its source and a \groupno~counter recording the number of settled robots. The $single$-source robot changes its state to $wait$ and remains at its current node until it is discovered and merged into a passing DFS traversal. Each $multi$-source initiates a bounded DFS traversal, where DFS heads traverse the graph, forming groups of size $[2c,4c]$. When a DFS head reaches a node, robots wait a fixed interval to detect collisions with other DFS traversals. Subsumption rules are applied deterministically, where the traversal is guided by the robots responsible for the groups, and the traversal with more settled robots survives, and the others are merged. Each robot in a subsumed traversal is reassigned to the surviving DFS, which may revisit each node of the collapsed traversal at most once. When a surviving DFS encounters a $single$-source robot, that robot is immediately merged into the traversal and becomes part of its settled group.

\begin{theorem}\label{thm:general-k=n}
\cd~can be solved deterministically in general initial configuration for $k=n$ with $k>1$ robots having prior knowledge of $n$ and $k$ in $O(n/k \cdot m)$ rounds using $O(n/k \cdot \log k)$ bits per robot.
\end{theorem}

\begin{proof}
Consider a connected graph $G=(V,E)$ with $|V|=n$ and $k=n>1$ robots in a general initial configuration, where each robot knows $n$ and $k$. Each robot initializes a \textit{treelabel} and a \groupno~counter equal to the number of settled robots at its source.

Each multi-source robot initiates a bounded DFS traversal, forming groups of size $[2c,4c]$. When two DFS traversals collide, a deterministic subsumption rule is applied: the traversal with more settled robots survives, and the other is merged. Since group sizes strictly increase, the number of active traversals decreases monotonically, ensuring termination. Single-source robots remain in the $wait$ state until discovered and are then merged into a passing DFS.

Each subsumed traversal may cause nodes to be revisited, but each node is revisited at most once. Thus, each DFS explores $O(m)$ edges, and since at most $O(n/k)$ DFS traversals survive sequentially, the total running time is $O(n/k \cdot m)$ rounds.

Each robot stores a treelabel, a group counter, and DFS state, requiring $O(n/k \cdot \log k)$ bits of memory.

Hence, \cd~is solved deterministically in general initial configurations within the stated time and memory bounds.
\end{proof}

\subsection{When \texorpdfstring{$k< n$}{k<n}}
We consider the setting where robots may or may not know $n$ and $k$. 

\textbf{Algorithm:} Each source $S_p$ initiates a DFS with an initial guess $n = k_p$, forming bounded-size groups coordinated by \osci~robots. The main challenge is that independent DFSs may never meet at arbitrary nodes due to differences in group sizes and unsynchronized oscillations between groups of different DFSs. Consequently, merging cannot rely on incidental encounters during exploration.  

To guarantee convergence, each source retains a $passive$ robot. Since each DFS traversal is guaranteed to eventually visit every node, any traversal reaching a source will encounter its $passive$ robot and merge with the traversal originating there. However, an additional challenge arises when sources contain exactly two robots.

Let us consider the source where exactly two robots are present at a source $s$. To guarantee convergence, one robot remains at $s$ in a $passive$ state while the other initiates exploration. However, a single robot cannot explore an anonymous graph, as discussed earlier in Section \ref{the:single case}. 

This leads us to two possible cases, depending on the initial configuration of sources.

\textbf{Case 1: If there exists at least one source initially hosting more than two robots:} Suppose there exists a source $S$ that initially contains at least three robots. In this setting, the DFS initiated at $S$ can proceed normally, leaving one $passive$ robot at the source. Whenever this DFS encounters another traversal or a $single$-source robot, merging occurs according to the standard subsumption rule: the traversal with more settled robots continues, and the smaller traversal is absorbed. As a result, all robots eventually gather at the dominant traversal's source and continue exploration as a single unified DFS.

If some other source initially contains exactly two robots, those two robots start a DFS without leaving behind a $passive$ robot, since a single robot cannot explore. Their traversal will eventually reach the source $S$, where the $passive$ robot remains, and will merge into the DFS originating from $S$. 

Therefore, the existence of even one source with at least three robots guarantees global convergence: all DFS processes eventually merge into one, which explores the entire graph and subsequently performs dispersion.

This scenario requires only $O(n/k \cdot\log (k + \Delta))$ bits of memory per robot, matching the memory bound of the general algorithm.

\textbf{Case 2: If there does not exist at least one source initially hosting more than two robots:} Consider a source $S$ that initially contains exactly two robots. Since a single robot cannot explore, both robots initiate independent bounded DFS traversals, leaving no $passive$ robot at the source. There are two possible cases:

\begin{itemize}
    \item If either of the two robots of source $S$ encounters another DFS traversal or a $single$-source robot, merging occurs according to the standard subsumption rule: the traversal containing more settled robots continues, while the smaller traversal is absorbed. Once merging results in a source with more than two robots, the process switches to Case 1: one robot becomes $passive$ at the source, and the remaining robots initiate a new DFS.

    \item If no encounters take place during the exploration, this situation indicates that every source (if it exists) initially contained exactly two robots and that the DFS traversals originating from different sources remained completely disjoint throughout their execution. Since none of the traversals interacted or merged with any other traversal or $single$-source robot, each pair of robots independently completes its bounded exploration and acquires full knowledge of the graph $G$. The two robots at the source $S$ then disperse to distinct nearest nodes based on a matching between node colors and robot colors. If a node is already occupied by another robot, a robot uses its stored knowledge to move to the next nearest available node. In the event that two or more robots reach the same node simultaneously, tie-breaking is performed using their unique IDs.

\end{itemize}

This process ensures successful dispersion, but it requires each robot to store complete graph knowledge, resulting in a memory requirement of $O(n \log (k + \Delta))$ bits per robot.





\begin{theorem}\label{thm:general-unknown}
\cd~can be solved deterministically in general initial configuration in $O(\log(n/k)\cdot n/k \cdot m)$ rounds. The algorithm requires (i) $O(n/k \cdot \log(k+\Delta))$ bits of memory per robot if some source initially hosts at least three robots, and (ii) $O(n \cdot \log(k+\Delta))$ bits of memory per robot otherwise.
\end{theorem}
\begin{proof}
    
We prove the theorem by contradiction. Assume, for the sake of contradiction, that the algorithm does not solve the \cd~problem under the stated conditions. Then at least one of the following must hold:

(i) the algorithm does not terminate, or  
(ii) it terminates but fails to achieve correct dispersion.

Assume that the algorithm does not terminate. Then infinitely many DFS phases are executed without achieving convergence.

Consider any DFS phase. Either at least two DFS traversals intersect, or they do not. If two DFS traversals intersect, then by the subsumption rule, exactly one traversal survives, and the other is absorbed. Hence, the total number of active DFS traversals strictly decreases. Since the number of sources is finite, this process cannot continue indefinitely. Therefore, infinitely many phases without termination is impossible in this case.

Otherwise, suppose that no DFS traversals intersect in a given phase. Then each traversal completes its bounded exploration and returns to its source. By the algorithm, the exploration capacity is increased exponentially in the next phase. Since the graph $G$ is finite, after finitely many phases, the exploration range of at least one traversal becomes large enough to intersect another traversal or reach a source containing a passive robot. This again forces subsumption and reduces the number of DFS traversals.

Thus, in every phase, either the number of traversals strictly decreases or the exploration capacity increases. Both quantities are bounded, implying that the algorithm must terminate. This contradicts the assumption of non-termination.

Assume now that the algorithm terminates but does not achieve dispersion. If the execution reaches a configuration with exactly one active DFS traversal, then the situation reduces to the rooted case. By correctness of the rooted algorithm (by Lemma \ref{sec:rooted-unknown}), dispersion is achieved, contradicting the assumption.

Otherwise, suppose the algorithm terminates with multiple DFS traversals that never merged. This can occur only if every source initially contains exactly two robots and no traversal ever intersects another. In this case, each DFS eventually completes its exploration and acquires complete knowledge of the graph. Using this information, each pair of robots disperses to distinct nodes via deterministic tie-breaking. Hence, dispersion is still achieved, contradicting the assumption.

Since both non-termination and incorrect termination lead to contradictions, the algorithm must terminate and correctly solve the \cd~problem.

Each DFS phase requires $O(n/k \cdot m)$ rounds, and the exploration capacity doubles in each phase, yielding $O(\log(n/k))$ phases. Thus, the total round complexity is $O(\log(n/k)\cdot n/k \cdot m)$. The memory bounds follow directly from whether robots store partial or complete graph information, as stated in the theorem. Hence, the theorem holds.
\end{proof}

\section{Arbitrary Graph: Dispersed, \texorpdfstring{$k\leq n$}{k<n}, Known \texorpdfstring{$n$}{n}}\label{sec:dispersed-known}

\subsection{Special Case: When \texorpdfstring{$k= n$}{k=n}  \label{sec:dispersed-k=n}}

 In this subsection, we study the \cd~problem under dispersed initial configurations, where each node of the graph initially hosts at most one robot. Thus, $k=n$ robots are placed on $n$ distinct nodes of an arbitrary graph $G$, and both $k$ and $n$ are known to all robots.

 Our approach builds upon the ``meeting protocol for two adjacent agents'' in \cite{chand2025brief}, which guarantees that two adjacent agents meet within $O(\log L)$ rounds, where $L$ is the ID of highest ID robot. Over time, the initially dispersed configuration evolves into a general configuration, in which at least one node hosts two or more robots.

 \textbf{\emph{Meeting protocol for two adjacent agents}:} Consider two robots $r_u$ and $r_v$ initially located at adjacent nodes $u$ and $v$ of the graph, where $u$ is connected to $v$ via an outgoing port $p_u$ at $u$. The goal is for robot $r_u$ to meet its neighbour $r_v$ through port $p_u$.  Each robot first pads its ID to a $\log L$-bit binary string and then appends its bitwise complement to the left, thereby forming a new ID of length $2 \log L$ bits, where $\log L$ is the length of the largest ID robot. The protocol proceeds bit by bit, from the least significant bit (LSB) to the most significant bit (MSB). For each bit position $i$, if the $i$-th bit of $r_u$'s new ID is $1$, the robot moves from node $u$ to node $v$ in the first round of the pair and returns to $u$ in the second; if the bit is $0$, the robot remains at $u$ for both rounds. Since all robots have distinct IDs, there exists at least one bit position at which one robot's new ID has a $1$ while the other has a $0$. In that round, one robot moves while the other remains still, ensuring that the two robots meet. 

 Thus, the protocol completes in $4 \log L$ rounds, with two rounds devoted to each bit position. The protocol is entirely local and requires no knowledge of the state or behaviour of $r_v$. A robot may invoke this procedure using any specific port it intends to traverse.

 Note that in \cite{chand2025brief}, the robots are assumed to have prior knowledge of $\log L$. In our setting, the robots know the value of $n$, and each robot with an ID in $[1,k^{O(1)}]$ and therefore is assumed to know its value.

 \textbf{Algorithm:} Initially, all robots are in a dispersed configuration, and each robot knows the value of $n$. Every robot $r$ positioned at a node $u$ initiates the \emph{Meeting Protocol for Two Adjacent Agents} by attempting to visit its neighbor $v$ through port 0. During the execution of this protocol, if robot $r$ meets another robot $r'$ at node $v$, it forms a group with $r'$ and thereafter follows $r'$, while $r'$ continues executing the algorithm. By repeatedly applying this protocol across all adjacent pairs, robots located on neighboring nodes are guaranteed to meet. After at most $4 \log k$ rounds, the initially dispersed configuration transforms into a general configuration in which at least one node hosts two or more robots.

 Once such a general configuration is obtained, we apply the algorithm developed for the general initial setting (Section~\ref{sec:general-known}), which guarantees deterministic \cd~for arbitrary placements with $k =n$.


\begin{theorem}\label{thm:dispersed_k=n}
\cd~can be solved deterministically in dispersed initial configuration with $k=n$ robots having prior knowledge of $n,k$ in $O(\log(k)+ n/k \cdot m)$ rounds using $O(n/k \cdot \log k)$ bits per robot.


\end{theorem}
\begin{proof}
We prove correctness, termination, and complexity of the algorithm.

Initially, the robots are in a dispersed configuration, with exactly one robot per node and with full knowledge of $n$ and $k=n$. Each robot independently initiates the \emph{Meeting Protocol for Two Adjacent Agents} with one of its neighbors. By the properties of the protocol, any two robots located at adjacent nodes are guaranteed to meet within $O(\log L)$ rounds.

Since the graph $G$ is connected, repeated executions of the meeting protocol ensure that at least one pair of adjacent robots meets, forming a group of size at least two. Once such a meeting occurs, the configuration is no longer dispersed and becomes a general configuration, in which at least one node hosts multiple robots.

After reaching a general configuration, the robots execute the deterministic algorithm for the general initial setting described in Section \ref{sec:general-known}. By Theorem \ref{thm:general-unknown}, this algorithm correctly solves the \cd~problem for arbitrary initial placements with known $n$, ensuring that all robots are placed on distinct nodes and that each node hosts at most one robot. Hence, correctness follows.


The meeting phase requires at most $O(\log L)$ rounds to transform the dispersed configuration into a general configuration. Once this occurs, the algorithm for the general setting requires $O(\log(n/k) \cdot n/k \cdot m)$ rounds. Since $k=n$ in the dispersed setting, the total round complexity is $O(\log(L)+ \log(n/k) \cdot n/k \cdot m)$. Hence, the time and memory complexity follow from Theorem \ref{thm:general-unknown}.
\end{proof}

\subsection{When \texorpdfstring{$k< n$}{k<n}}

We study the \cd~problem under dispersed initial configurations for $k<n$ robots, where each node of the graph initially hosts at most one robot. Thus, $k$ robots occupy $k$ distinct nodes of an arbitrary graph $G$, and both $k$ and $n$ are known to all robots.

The high-level idea is to transform the dispersed configuration into a general configuration, i.e., a configuration in which at least one node hosts two or more robots. Once such a source exists, the \cd~problem becomes solvable using the techniques developed for general configurations.

This transformation is closely related to the classical deterministic gathering problem. In particular, the work \cite{molla2023fast} addresses the setting in which $k$ robots are placed on the nodes of an $n$-node graph, with the objective of having all $k$ robots meet at a single node and terminate. They present deterministic algorithms that achieve gathering with detection on arbitrary graphs.

\paragraph{Summary of results from prior work \cite{molla2023fast}:} The authors first present a universal exploration sequence (UXS)-based gathering algorithm that requires $\tilde{O}(n^5)$ rounds. They then focus on undispersed configurations (where some nodes initially host multiple robots), and propose a more efficient two-phase algorithm. In the first phase, a group of robots collaboratively constructs a map of the graph in $O(n^3)$ rounds, while in the second phase, this map is used to route all robots to a designated gathering node.

They then address the dispersed configuration by first attempting to convert it into an undispersed configuration. The resulting time complexity depends on the minimum distance between any two robots. If the closest pair of robots is at distance $i$, the robots meet and form an undispersed configuration in $O(n^i \log n)$ rounds. Subsequently, all robots gather at a designated node in $\tilde{O}(n^5)$ rounds.

Hence, the following complexity trade-offs hold for gathering all robots at a single node from a dispersed configuration:
\begin{itemize}
    \item The gathering completes in $O(n^3)$ rounds if the initial configuration is undispersed, or if it is dispersed with at least two robots located at distance $2$ from each other.
    \item The gathering completes in $O(n^i \log n)$ rounds if the initial configuration is dispersed with at least two robots located at distance $i$, for $i \in \{3,4,5\}$.
    \item Otherwise, the gathering requires $\tilde{O}(n^5)$ rounds.
\end{itemize}

Consequently, the following bounds hold as a function of the number of robots $k$, for arbitrary initial configurations:
\begin{enumerate}
    \item If $k \ge \lfloor n/2 \rfloor + 1$, the gathering completes in $O(n^3)$ rounds.
    \item If $k \ge \lfloor n/3 \rfloor + 1$, the gathering completes in $O(n^4 \log n)$ rounds.
    \item Otherwise, the gathering requires $\tilde{O}(n^5)$ rounds.
\end{enumerate}

Importantly, the algorithm assumes knowledge of the graph size $n$ but does not require knowledge of the number of robots $k$.

\textbf{Algorithm:} Our approach builds upon the results of \cite{molla2023fast}. Since the robots know $k$ or $n$, they can gather at a single node by directly applying the deterministic gathering routines from that work. In particular, if the closest pair of robots is at distance $i \in \{1,2,3,4,5\}$, then the gathering completes in $O(n^{i} \log n)$ rounds. Once such a general configuration is obtained, we invoke the algorithm designed for the general initial setting (Section \ref{sec:general-known}).

If no such pair of robots exists, the gathering procedure brings all robots to a single node in $\tilde{O}(n^{5})$ rounds. After this rooted configuration is formed, we apply the algorithm developed for the rooted initial setting (Section \ref{sec:rooted-known}).

\begin{theorem}\label{thm:dispersed_k<n}
The \cd~problem can be solved deterministically in a dispersed initial configuration with $k<n$ robots and prior knowledge of $n$, using $O(M^{*} + n/k\cdot \log(k+\Delta))$ bits of memory per robot, where $M^{*}$ denotes the memory required to implement a Universal Exploration Sequence (UXS). The time complexity is: (i) $O(n^3 + n/k\cdot m)$ rounds if $i=1,2$, where $i$ is the distance between the closest pair of robots; (ii) $O(n^{i}\log n + n/k\cdot m)$ rounds if $i \in \{3,4,5\}$; (iii) $\tilde{O}(n^5)$ rounds otherwise.
\end{theorem}

\begin{proof}

Initially, the robots are in a dispersed configuration with $k<n$ robots placed on distinct nodes of a connected graph $G$, and all robots know $n$ and $k$. The goal is to transform this dispersed configuration into a configuration for which the \cd~problem is already solvable.

The algorithm first applies a deterministic gathering procedure from \cite{molla2023fast}. This procedure guarantees that all robots gather at a single node or that at least one node hosts multiple robots, depending on the distance between the closest pair of robots. In either case, the resulting configuration is no longer dispersed.

If the gathering procedure results in a general configuration, that is, a configuration in which at least one node hosts two or more robots, we invoke the algorithm for general initial configurations (Section~\ref{sec:general-known}). By Theorem~\ref{thm:general-unknown}, this algorithm correctly solves the \cd~problem.

If the gathering procedure results in a rooted configuration, in which all robots gather at a single node, we apply the algorithm for rooted initial configurations (Section~\ref{sec:rooted-known}). By the correctness of that algorithm, dispersion is achieved deterministically.

Thus, in all cases, the algorithm produces a correct colorful dispersion.

From \cite{molla2023fast}, if the closest pair of robots is at distance $i \in \{2,3,4,5\}$, the gathering phase completes in $O(n^{i}\log n)$ rounds. If no such pair exists, the gathering phase completes in $\tilde{O}(n^{5})$ rounds. Once the gathering phase terminates, the resulting configuration is either 
rooted or general. In both cases, the subsequent execution preserves the asymptotic bounds, and thus the claimed time and memory complexities follow.
\end{proof}

\section{Arbitrary Graph: Unknown, \texorpdfstring{$k\leq n$}{k<n}}\label{sec:unknown-confi}

For known $n$, if the configuration is \emph{rooted} and $k$ is known, the robots can immediately detect this and solve the problem using the algorithm of Section~\ref{sec:rooted-known}.  
For all other situations, the robots first gather at a single node, thereby creating a rooted configuration, using the result of Molla {\it et al.}~\cite{molla2023fast}, deterministic gathering with detection is possible on arbitrary graphs using only the knowledge of $n$, via universal exploration sequences. 
When the initial configuration is already rooted but $k$ is unknown, the robots explore the graph to enable correct identification. Once a rooted configuration is formed and $n$ is known, the robots apply the rooted algorithm (Section~\ref{sec:rooted-known}) to solve  \cd.

If unknown $n$ but $k$ is known, robots can still detect a rooted configuration. They execute the rooted algorithm using exponential guessing of $n$, as described in Section~\ref{sec:rooted-unknown}.  
The remaining configurations, where neither $n$ nor the configuration type is known, remains as an open problem.
Recall that, no initial configuration (known or unknown) as well as known/unknown $n,k$ denied solution for \dpn.

\begin{theorem}\label{thm:unknown-config}
\cd~can be solved deterministically on an arbitrary graph $G$ with $k\le n$ robots, where the initial configuration is unknown but the robots have prior knowledge of $n$, using $O(M^{*} + n/k\cdot \log(k+\Delta))$ bits of memory per robot, where $M^{*}$ denotes the memory required to implement a UXS. The time complexity is: (i) $O(n^{3} + n/k\cdot m)$ rounds if $k \ge \lfloor n/2 \rfloor + 1$; (ii)  $O(n^{4}\log n + n/k\cdot m)$ rounds if $k \ge \lfloor n/3 \rfloor + 1$; (iii) $\tilde{O}(n^5)$ rounds otherwise.
\end{theorem}

\begin{proof}
We prove that \cd~can be solved deterministically when the robots have prior knowledge of $n$, regardless of the initial configuration.

If the initial configuration is rooted and $k$ is known, the robots immediately execute the rooted algorithm of Section~\ref{sec:rooted-known}, which guarantees correct dispersion.

Otherwise, the robots first apply the deterministic gathering-with-detection algorithm of Molla \textit{et al.}~\cite{molla2023fast}, which works on arbitrary graphs using only the knowledge of $n$ and universal exploration sequences. Upon termination, all robots are gathered at a single node and can detect this event, thereby creating a rooted configuration.

If the configuration is rooted but $k$ is unknown, the robots explore the graph to enable correct identification of the rooted configuration. Once rootedness is confirmed and $n$ is known, they execute the rooted algorithm of Section~\ref{sec:rooted-known}.

If $k$ is known but $n$ is unknown, the robots detect rootedness and execute the rooted algorithm using exponential guessing of $n$, as described in Section~\ref{sec:rooted-unknown}. Correctness follows since the guessed value eventually exceeds the true $n$.

After a rooted configuration is obtained with a correct value of $n$, dispersion is completed by the rooted algorithm, settling exactly one robot at each occupied node.

The memory required per robot is $O(M^{*} + n/k \cdot \log(k+\Delta))$, where $M^{*}$ is the memory needed to implement a universal exploration sequence. The stated time bounds follow from the complexities of deterministic gathering and rooted dispersion in the respective cases.

Hence, \cd~can be solved deterministically on an arbitrary graph with unknown initial configuration when $n$ is known.
\end{proof}



\section{Conclusion and Discussion}

 In this work, we studied the \cd~problem under rooted, dispersed, and general configurations, considering different assumptions on the robots’ prior knowledge of $n$ and $k$. The results are summarized in Table~\ref{table:table-result}. 

 We also examined the more challenging setting in which the initial configuration is unknown. When robots know $k$, they can determine whether the configuration is rooted. However, in the absence of a rooted configuration or without knowledge of $k$, a single robot cannot distinguish between dispersed and general configurations, resulting in several open cases for unknown initial configurations. Furthermore, once a dispersed configuration is achieved, the \cd~problem becomes solvable for both rooted and general configurations. These findings highlight the intrinsic connection between configuration identification, graph exploration, and the solvability of colorful dispersion.

 An important direction for future work is to resolve the remaining open cases by identifying minimal additional assumptions, such as limited node memory, partial knowledge of graph parameters, or the availability of additional robots, that enable configuration identification and graph exploration. Another important direction is to see whether time and/or memory complexities of our proposed algorithms could be improved to match the respective time/memory lower bounds. Closing these gaps, as well as extending our results to randomized or fault-tolerant settings, constitute natural avenues for future research.

\bibliographystyle{IEEEtranS}

\end{document}